\documentclass[a4paper, 11pt]{article}
\usepackage{fullpage}
\usepackage{graphicx}
\usepackage{amsmath, amsthm, amssymb, comment, url}
\usepackage[unicode]{hyperref}
\usepackage{booktabs}

\newcounter{rnc}
\newcommand{\Rn}[1]{\setcounter{rnc}{#1}\Roman{rnc}}

\newtheorem{theorem}{Theorem}[section]
\newtheorem{lemma}[theorem]{Lemma}

\newtheorem{claim}[theorem]{Claim}
\newtheorem{observation}[theorem]{Observation}
\newtheorem{corollary}[theorem]{Corollary}
\theoremstyle{definition}
\newtheorem{definition}[theorem]{Definition}

\theoremstyle{remark}

\newcommand{\CC}{{\mathbb C}}
\newcommand{\QQ}{{\mathbb Q}}
\newcommand{\RR}{{\mathbb R}}
\newcommand{\ZZ}{{\mathbb Z}}
\newcommand{\bM}{{\bf M}}
\newcommand{\tbM}{\tilde{\bf M}}
\newcommand{\cF}{{\cal F}}
\newcommand{\cI}{{\cal I}}
\newcommand{\tcI}{\tilde{\cal I}}
\newcommand{\cL}{{\cal L}}
\newcommand{\cS}{{\cal S}}

\newcommand{\cX}{{\cal X}}
\newcommand{\cY}{{\cal Y}}
\newcommand{\cZ}{{\cal Z}}
\newcommand{\tm}{\tilde{m}}

\newcommand{\tE}{\tilde{E}}
\newcommand{\tF}{\tilde{F}}
\newcommand{\tG}{\tilde{G}}

\newcommand{\tM}{\tilde{M}}

\newcommand{\tV}{\tilde{V}}
\newcommand{\hE}{\hat{E}}

\newcommand{\hH}{\hat{H}}

\newcommand{\hU}{\hat{U}}

\title{Making Bipartite Graphs DM-irreducible}
\author{
  Krist\'{o}f B\'{e}rczi\thanks{
    E\"otv\"os Lor\'and University, 1117 Budapest, Hungary.
    Email: \texttt{berkri@cs.elte.hu}
  } \and Satoru Iwata\thanks{
    University of Tokyo, Tokyo 113-8656, Japan.
    Email: \texttt{iwata@mist.i.u-tokyo.ac.jp}
  } \and Jun Kato\thanks{
    Toyota Motor Corporation, Toyota 471-8571, Japan.
    Email: \texttt{jun\_kato\_aa@mail.toyota.co.jp}
  } \and Yutaro Yamaguchi\thanks{
    Osaka University, Osaka 565-0871, Japan.
    Email: \texttt{yutaro\_yamaguchi@ist.osaka-u.ac.jp}
  }}
\date{\empty}

\begin{document}
\maketitle
\thispagestyle{empty}

\begin{abstract}
The Dulmage--Mendelsohn decomposition (or the DM-decomposition)
gives a unique partition of the vertex set of a bipartite graph
reflecting the structure of all the maximum matchings therein.
A bipartite graph is said to be DM-irreducible if its DM-decomposition
consists of a single component.

In this paper, we focus on the problem of making a given bipartite graph
DM-irreducible by adding edges.
When the input bipartite graph is balanced
(i.e., both sides have the same number of vertices)
and has a perfect matching,
this problem is equivalent to making a directed graph
strongly connected by adding edges,
for which the minimum number of additional edges was characterized by Eswaran and Tarjan (1976).

We give a general solution to this problem,
which is divided into three parts.
We first show that our problem can be formulated as a special case 
of a general framework of covering supermodular functions,
which was introduced by Frank and Jord\'an (1995) to investigate
the directed connectivity augmentation problem.
Secondly, when the input graph is not balanced,
the problem is solved via matroid intersection.
This result can be extended to the minimum cost version
in which the addition of an edge gives rise to an individual cost.
Thirdly, for balanced input graphs,
we devise a combinatorial algorithm that
finds a minimum number of additional edges to attain the DM-irreducibility,
while the minimum cost version of this problem is NP-hard.
These results also lead to min-max characterizations of the minimum number,
which generalize the result of Eswaran and Tarjan.
\end{abstract}

\clearpage
\thispagestyle{empty}
\tableofcontents
\clearpage
\setcounter{page}{1}

\section{Introduction}\label{sec:Introduction}
The {\em Dulmage--Mendelsohn decomposition}~\cite{DM1958, DM1959}
(or the {\em DM-decomposition})
of a bipartite graph gives a unique partition of the vertex set,
which reflects the structure of all the maximum matchings therein
(see Section~\ref{sec:DM-decomposition} for the details).
A bipartite graph is said to be {\em DM-irreducible}
if its DM-decomposition consists of only one nonempty component.

In this paper, we focus on the following question:
how many additional edges are necessary
to make a given bipartite graph $G$ DM-irreducible?
\begin{description}
  \setlength{\itemsep}{0mm}
\item[\underline{Problem (DMI)}]

\item[Input:]\vspace{1mm}
  A bipartite graph $G = (V^+, V^-; E)$.

\item[Goal:]
  Find a minimum-cardinality set $F$ of additional edges
  such that $G + F$ is DM-irreducible.
\end{description}\vspace{1mm}

Throughout this paper, for an input bipartite graph $G = (V^+, V^-; E)$,
we define $n := \max\{|V^+|, |V^-|\}$, $\ell := \min\{|V^+|, |V^-|\}$, and $m := |E|$.
We say that $G$ is {\em balanced} if $n = \ell$, and {\em unbalanced} otherwise.
We denote by ${\rm opt}(G)$ the optimal value of Problem (DMI),
i.e., the minimum number of additional edges to make $G$ DM-irreducible.

\medskip
When $G$ is balanced and has a perfect matching,
Problem (DMI) is equivalent to the problem of making a directed graph strongly connected
by adding as few edges as possible (see Section~\ref{sec:Equivalence_SC}).
Eswaran and Tarjan \cite{ET1976} introduced the latter problem,
and gave a simple solution (Theorem~\ref{thm:ET}).

A natural generalization of the strong connectivity augmentation is to find
a smallest set of additional edges that make a given directed graph $k$-connected
(i.e., such that removing at least $k$ vertices is needed to violate strong connectivity).
In order to investigate this problem, Frank and Jord\'an
\cite{FJ1995} introduced a general framework of covering a crossing supermodular
function by directed edges. They provided a min-max duality theorem and
a polynomial-time algorithm relying on the ellipsoid method.
Later, V\'egh and Bencz\'ur \cite{VB2008} devised a combinatorial algorithm whose
running time bound is pseudopolynomial, depending polynomially on the function values.

In this paper, we give a general solution to Problem (DMI) summarized as follows.
\begin{itemize}
  \setlength{\itemsep}{.5mm}
\item
  In general, the problem is within the Frank--Jord\'an framework.
\item
  When $G$ is unbalanced, the problem is solved via matroid intersection.
\item
  When $G$ is balanced, the problem is directly solved by an efficient algorithm.
\end{itemize}

\subsection{Summary of main results}
We first show that Problem (DMI) is a special case of the Frank--Jord\'an framework in general.
To be precise, we reduce the unbalanced case to the balanced case,
and then formulate the balanced case in terms of the Frank--Jord\'an framework.
As a main consequence of this reduction,
we derive the following min-max duality on Problem (DMI)
from the min-max duality theorem of Frank and Jord\'an.

For a one-side vertex set $X \subseteq V^\pm$ in a bipartite graph $G = (V^+, V^-; E)$,
we denote by $\Gamma_G(X) \subseteq V^\mp$ the set of vertices in the other side
that are adjacent to some vertex in $X$.
For a set $S$, a {\em subpartition} of $S$
is a partition of some subset of $S$ (i.e., a family of disjoint nonempty subsets of $S$).
A subpartition $\cX$ of $S$ is said to be {\em proper} if $\cX \neq \{S\}$.
For a subpartition $\cX$ of $V^+$ or of $V^-$, we define
\begin{equation}\label{eq:duality}
  \tau_G(\cX) := \sum_{X \in \cX} \left(|X| - |\Gamma_G(X)| + 1\right).
\end{equation}
Recall that ${\rm opt}(G)$ denotes the optimal value of Problem {\rm (DMI)}.

\begin{theorem}\label{thm:duality}
  Let $G = (V^+, V^-; E)$ be a bipartite graph with $|V^+| = |V^-| \geq 2$.
  Then we have
  \[{\rm opt}(G) = \max_\cX\tau_G(\cX),\]
  where the maximum is taken over all proper subpartitions $\cX$ of $V^+$ and of $V^-$.
\end{theorem}

\begin{theorem}\label{thm:duality_unbalanced}
  Let $G = (V^+, V^-; E)$ be a bipartite graph with $|V^+| < |V^-|$.
  Then we have
  \[{\rm opt}(G) = \max_{\cX^+}\tau_G(\cX^+),\]
  where the maximum is taken over all subpartitions $\cX^+$ of $V^+$.
\end{theorem}

Besides, the function values that appear in the reduction to the Frank--Jord\'an framework
are bounded by ${\rm O}(n)$,
and hence a direct application of the V\'egh--Bencz\'ur algorithm runs in polynomial time.
Although this reduction reveals the tractability of Problem (DMI),
the running time is not satisfactory.
Similarly to the directed connectivity augmentation, it requires ${\rm O}(n^7)$ time.
As seen below,
the Frank--Jord\'an framework is in fact excessively generalized to handle our problem,
and one can solve it much more simply and efficiently
(cf. Theorems~\ref{thm:algorithm_unbalanced} and \ref{thm:algorithm}).

\medskip
As the second result,
we show that the unbalanced case reduces to the matroid intersection problem.
Then, with the aid of a fast matroid intersection algorithm,
one can solve the unbalanced case in ${\rm O}(n + m\sqrt{\ell} \log \ell)$ time.

\begin{theorem}\label{thm:algorithm_unbalanced}
  For a bipartite graph $G = (V^+, V^-; E)$ with $\ell = |V^+| < |V^-| = n$ and $|E| = m$,
  one can find in ${\rm O}(n + m\sqrt{\ell} \log \ell)$ time a minimum number of additional edges
  to make $G$ DM-irreducible.
\end{theorem}

Our reduction to matroid intersection can be utilized
even when the addition of each edge gives rise to an individual cost
and we are required to minimize the total cost. 
By using a weighted matroid intersection algorithm,
one can solve the minimum-cost augmentation problem in ${\rm O}(n^2\ell)$ time.
In contrast, in the balanced case, the minimum-cost augmentation is NP-hard
even when $G$ has a perfect matching and the number of different cost values is at most two
(which was shown in \cite{ET1976} for the strong connectivity augmentation).
These facts imply that there is a significant gap of the difficulty of the weighted versions
between the balanced and unbalanced cases.

In addition,
we derive the min-max duality for the unbalanced case (Theorem~\ref{thm:duality_unbalanced})
from Edmonds' matroid intersection theorem,
while it can be shown via the reduction 
to the balanced case by using Theorem~\ref{thm:duality} (see Appendix~\ref{sec:app2}).

\medskip
The third result is a direct combinatorial algorithm
for the balanced case of Problem (DMI), which runs in ${\rm O}(nm)$ time.
While the unbalanced case is efficiently solved via matroid intersection,
one can also use this algorithm to solve the unbalanced case
with the aid of the reduction to the balanced case (see Appendix~\ref{sec:unbalanced}).

\begin{theorem}\label{thm:algorithm}
  For a bipartite graph $G = (V^+, V^-; E)$ with $|V^+| \leq |V^-| = n$ and $|E| = m$,
  one can find in ${\rm O}(nm)$ time a minimum number of additional edges
  to make $G$ DM-irreducible.
\end{theorem}

Our algorithm also gives an alternative proof of Theorem~\ref{thm:duality},
which is constructive in the sense that
one can easily construct a maximizer of $\tau_G$ as an optimality certificate
when the algorithm halts.
It is worth mentioning that one can maximize $\tau_G$-like functions
in polynomial time in a more general situation (see Appendix \ref{sec:app1}).

\subsection{Related work}
For the directed $k$-connectivity augmentation,
which is also within the Frank--Jord\'an framework,
Frank and V\'egh \cite{FV2008} gave a much simpler combinatorial algorithm
when a given directed graph is already $(k-1)$-connected.
Since ``$0$-connected'' enforces no constraint and ``$1$-connected'' is equivalent to ``strongly connected,''
this special case also generalizes the strong connectivity augmentation.
The direction of generalization is, however, different from our problem.
The Frank--V\'egh setting is translated in terms of bipartite graphs as follows:
for a given {\em $(k-1)$-elementary} balanced bipartite graph $G$,
to make $G$ {\em $k$-elementary} by adding a minimum number of edges,
where ``$0$-elementary'' and ``$1$-elementary'' are equivalent to ``perfectly matchable''
and to ``DM-irreducible,'' respectively,
and ``$k$-elementary'' is strictly stronger than ``DM-irreducible'' when $k \geq 2$.
In our problem, we are required to make a balanced bipartite graph $G$ $1$-elementary
even when $G$ is not $0$-elementary.

The DM-decomposition is known to be a useful tool
in numerical linear algebra (see, e.g., \cite{DER1986}).
A bipartite graph associated with a matrix is naturally defined by its nonzero entries,
and its DM-decomposition 
gives the finest block-triangularization,
which helps us to solve the system of linear equations efficiently. 
The finer decomposed, the finer from computational point of view.
Hence the DM-irreducibility is not a desirable property in this context.

There are, however, certain situations in which DM-irreducibility is rather preferable. 
For example, in game theory, the uniqueness of the utility profile
in a subgame perfect equilibrium
in a bargaining game is characterized by DM-irreducibility. In control theory, 
the structural controllability is characterized in terms of DM-irreducibility. 
We explicate these situations and possible applications of our result 
in Section~\ref{sec:Applications}.

\subsection{Organization}
The rest of this paper is organized as follows.
In Section~\ref{sec:Preliminaries},
we describe necessary definitions and known results
on the DM-decomposition of bipartite graphs
and on the strong connectivity of directed graphs.
In Section~\ref{sec:FJ}, we reduce the general case of Problem (DMI)
to supermodular covering framework of Frank and Jord\'an,
and apply their result to prove Theorem~\ref{thm:duality}.
In Section~\ref{sec:WMI},
we solve the unbalanced case via matroid intersection.
Section~\ref{sec:Algorithm} is devoted to presenting our direct algorithm for the balanced case.
The correctness of the algorithm
also gives an alternative, constructive proof of Theorem~\ref{thm:duality}.
A key procedure in our algorithm is shown separately in Section~\ref{sec:EPM}.
Finally, in Section~\ref{sec:Applications},
we discuss possible applications of our result in game theory 
and in control theory.

\section{Preliminaries}\label{sec:Preliminaries}
\subsection{Strong connectivity of directed graphs}\label{sec:Strong_connectivity}
Let $G = (V, E)$ be a directed graph.
A sequence $P = (v_0, e_1, v_1, e_2, v_2, \ldots, e_l, v_l)$
is called a {\em path} (or, in particular, a {\em $v_0$--$v_l$ path}) in $G$
if $v_0, v_1, \ldots, v_l \in V$ are distinct
and $e_i = v_{i-1}v_i \in E$ for each $i \in \{1, 2, \ldots, l\}$.
For two vertices $u, w \in V$ (possibly $u = w$),
we say that $u$ is {\em reachable to $w$}
(or, equivalently, $w$ is {\em reachable from $u$}) in $G$
and denote by $u \xrightarrow{G} w$
if there exists a $u$--$w$ path in $G$.
A directed graph is said to be {\em strongly connected} if
every two vertices are reachable to each other (also from each other).
A {\em strongly connected component} of $G$
is a maximal induced subgraph of $G$ that is strongly connected.
The strongly connected components of a directed graph
can be found in linear time with the aid of the depth first search \cite{Tarjan1972}.

Let $\cS = \{V_1, V_2, \ldots, V_k\}$ be the partition of $V$
according to the strongly connected components of $G$,
i.e., for any two vertices $u, w \in V$,
we have $u \xrightarrow{G} w$ and $w \xrightarrow{G} u$
if and only if $\{u, w\} \subseteq V_i$ for some $i$.
For $V_i, V_j \in \cS$,
we denote by $V_i \succeq_G V_j$
if $u \xrightarrow{G} w$
for every pair of $u \in V_i$ and $w \in V_j$.
Then the binary relation $\succeq_G$ is a partial order on $\cS$.
A strongly connected component of $G$
is called a {\em source component}
if its vertex set $V_i$ is maximal with respect to $\succeq_G$
(i.e., there is no $V_j \in \cS \setminus \{V_i\}$ with $V_j \succeq_G V_i$),
and a {\em sink component} if minimal.
Note that a strongly connected component is a source or sink component
if and only if no edge enters or leaves it, respectively.
The numbers of source and sink components of $G$
are denoted by $s(G)$ and $t(G)$, respectively.

Eswaran and Tarjan \cite{ET1976} characterized the minimum number of additional edges
to make a directed graph strongly connected, and proposed a linear-time algorithm
for finding such additional edges as follows.

\begin{theorem}[Eswaran--Tarjan {\cite[Section~2]{ET1976}}]\label{thm:ET}
  Let $G = (V, E)$ be a directed graph that is not strongly connected.
  Then the minimum number of additional edges
  to make $G$ strongly connected is equal to $\max\{s(G), t(G)\}$.
  Moreover, one can find such additional edges in ${\rm O}(|V| + |E|)$ time.
\end{theorem}

\subsection{DM-decomposition of bipartite graphs}\label{sec:DM-decomposition}
Let $G = (V^+, V^-; E)$ be a bipartite graph
with the vertex set $V$ partitioned into
the {\em left side} $V^+$ and the {\em right side} $V^-$.
Throughout this paper, a bipartite graph is dealt with as a directed graph
in which each edge is directed from left to right, i.e., $E \subseteq V^+ \times V^-$.
An edge set $M \subseteq E$ is called a {\em matching} in $G$
if $|\partial^+M| = |\partial^-M| = |M|$,
where $\partial^+M := \{\, u \mid uw \in M \,\} \subseteq V^+$
and $\partial^-M := \{\, w \mid uw \in M \,\} \subseteq V^-$.
A matching $M$ is said to be {\em maximum} if $|M|$ is maximum,
and {\bf perfect} if $|M| = \min\{|V^+|, |V^-|\}$
(this definition of ``perfect matchings'' is unusual,
where it extends a usual definition for the balanced bipartite graphs
to all the bipartite graphs).
A bipartite graph is said to be {\bf perfectly matchable}
if it has a perfect matching,
and {\bf matching covered} if every edge is contained in some perfect matching.

The DM-decomposition of a bipartite graph gives a unique partition of the vertex set,
which reflects the structure of all the maximum matchings therein as follows.
For a nonnegative integer $k$, we define $[k] := \{1, 2, \ldots, k\}$.
For a vertex set $X \subseteq V$,
we define $X^+ := X \cap V^+$ and $X^- := X \cap V^-$,
and denote by $G[X]$ the subgraph of $G$ induced by $X$.

\begin{theorem}[Dulmage--Mendelsohn~\cite{DM1958, DM1959}]\label{thm:DM}
  Let $G = (V^+, V^-; E)$ be a bipartite graph.
  Then there exists a partition $(V_0; V_1, V_2, \ldots, V_k; V_\infty)$ of
  $V$ such that\vspace{-.5mm}
  \begin{enumerate}
    \renewcommand{\labelenumi}{{\rm \arabic{enumi}.}}
    \setlength{\itemsep}{.5mm}
  \item
    either $|V_0^+| > |V_0^-|$ or $V_0 = \emptyset$,
  \item
    $G[V_i]$ is balanced $($i.e., $|V_i^+| = |V_i^-| > 0$$)$ and connected for each $i \in [k]$,
  \item
    either $|V_\infty^+| < |V_\infty^-|$ or $V_\infty = \emptyset$,
  \item
    $G[V_i]$ is matching covered for each $i \in [k] \cup \{0, \infty\}$, and
  \item
    every maximum matching in $G$ is a union of perfect matchings in $G[V_i]$.
  \end{enumerate}
\end{theorem}

We here define the {\em DM-decomposition} $(V_0; V_1, V_2, \ldots, V_k; V_\infty)$
of a bipartite graph $G = (V^+, V^-; E)$,
which satisfies the conditions in Theorem~\ref{thm:DM}
(see also, e.g., \cite{LP1986, Murota2000}).
Define a set function $f_G \colon 2^{V^+} \to \ZZ$ by
\begin{equation}\label{eq:deficiency}
  f_G(X^+) := |\Gamma_G(X^+)| - |X^+| \quad (X^+ \subseteq V^+),
\end{equation}
where recall $\Gamma_G(X^+) = \{\, w \mid \exists e = uw \in E \colon u \in X^+ \,\} \subseteq V^-$.
It is well-known that $f_G$ is submodular,
and hence all the minimizers of $f_G$ form a distributive lattice $\cL(f_G)$
with respect to the set union and intersection (see, e.g., \cite[Lemma 2.1]{Fujishige2005}).
For a maximal monotonically increasing sequence (called a maximal {\em chain})
$X_0^+ \subsetneq X_1^+ \subsetneq \cdots \subsetneq X_k^+$ in $\cL(f_G)$,
define $V_i := V_i^+ \cup V_i^-$ for each $i \in [k] \cup \{0, \infty\}$ as follows:
\begin{align*}
  &V_0^+ := X^+_0, &&V_0^- := \Gamma_G(X_0^+),\\[1mm]
  &V_i^+ := X_i^+ \setminus X_{i-1}^+, &&V_i^- := \Gamma_G(X_i^+) \setminus \Gamma_G(X_{i-1}^+) &(i \in [k]),\\[1mm]
  &V_\infty^+ := V^+ \setminus X_k^+, &&V_\infty^- := V^- \setminus \Gamma_G(X_k^+).
\end{align*}

It is known
that the resulting partition of $V$ with the following partial order $\sqsubseteq$
is unique (i.e., does not depend on the choice of a maximal chain in $\cL(f_G)$):
\[V_i \sqsubseteq V_j \iff \left[V_j^+ \subseteq X^+ \in \cL(f_G) \implies V_i^+ \subseteq X^+\right] \quad (i, j \in [k] \cup \{0, \infty\}).\]
Moreover, while $V^+$ and $V^-$ do not seem symmetric in the above definition,
it is also known that essentially the same partially-ordered partition is obtained by interchanging
the roles of $V^+$ and of $V^-$, in which, e.g., $V_0$ and $V_\infty$ are interchanged and
the direction of $\sqsubseteq$ is reversed.

\medskip
The DM-decomposition is known to be obtained as follows (cf. \cite[Section 2.2.3]{Murota2000}).
Take an arbitrary maximum matching $M \subseteq E$ in $G$.
Construct the {\em auxiliary graph} $G(M) := G + \overline{M}$ with respect to $M$,
where $\overline{M} := \{\, \bar{e} := wu \mid e = uw \in M \,\} \subseteq V^- \times V^+$
denotes the set of reverse edges.
The set of vertices reachable from some vertex in $V^+ \setminus \partial^+M$
in $G(M)$ is $V_0$,
and the set of vertices reachable
to some vertex in $V^- \setminus \partial^-M$ in $G(M)$ is $V_\infty$.
The rest $V_\ast := V \setminus (V_0 \cup V_\infty)$ is partitioned according to
the strongly connected components of $G_\ast := G(M)[V_\ast]$.
The partial order $\sqsubseteq$ is defined by $\preceq_{G_\ast}$
on $\{\, V_i \mid i \in [k] \,\}$ and so that $V_0$ and $V_\infty$
are minimum and maximum elements, respectively.
By this computation,
one can easily see the following properties.

\begin{observation}\label{obs:source_sink}
  Let $(V_0; V_1, V_2, \ldots, V_k; V_\infty)$ be the DM-decomposition
  of a bipartite graph $G = (V^+, V^-; E)$.
  Then, for any maximum matching $M \subseteq E$ in $G$,
  the auxiliary graph $G(M)$ satisfies the following conditions.\vspace{-.5mm}
  \begin{itemize}
    \setlength{\itemsep}{.5mm}
  \item
    No edge leaves $V_0$.
  \item
    No edge enters $V_\infty$.
  \item
    Each source component of $G(M)[V_0]$ is 
    a single vertex in $V^+ \setminus \partial^+ M$, and vice versa. 
    Hence, $s(G(M)[V_0]) = |V^+| - |M|$.
  \item
    Each sink component of $G(M)[V_\infty]$ 
    is a single vertex in $V^- \setminus \partial^- M$,
    and vice versa. Hence, $t(G(M)[V_\infty]) = |V^-| - |M|$.
  \end{itemize}
\end{observation}

A bipartite graph $G = (V^+, V^-; E)$ is said to be {\em DM-irreducible}
if its DM-decomposition consists of only one nonempty component,
i.e., either $V_0 = V$, $V_1 = V$, or $V_\infty = V$.
By the symmetry, we always assume $|V^+| \leq |V^-|$ without notice.
That is, if $G$ is unbalanced, then $|V^+| < |V^-|$.

\subsection{Relation to strong connectivity augmentation}\label{sec:Equivalence_SC}
From the computation of the DM-decomposition,
a balanced bipartite graph $G = (V^+, V^-; E)$ is DM-irreducible
if and only if $G$ has a perfect matching $M \subseteq E$ and
the auxiliary directed graph $G(M) = G + \overline{M}$ is strongly connected.
In addition,
a directed graph $G = (V, E)$ is strongly connected
if and only if the balanced bipartite graph $\tG = (\tV^+, \tV^-; \tE)$ defined as follows
is DM-irreducible:
\begin{align*}
  &\tV^+ := \{\, v^+ \mid v \in V \,\}, \quad \tV^- := \{\, v^- \mid v \in V \,\},\\[1mm]
  &\tE := \{\, u^+w^- \mid uw \in E \,\} \cup \{\, v^+v^- \mid v \in V \,\}.
\end{align*}
Note that $\tG$ has a perfect matching $\tM := \{\, v^+v^- \mid v \in V \,\} \subseteq \tE$,
and the DM-irreducibility of $\tG$ is equivalent to the strong connectivity of $\tG(\tM)$,
in which the two vertices $v^+ \in \tV^+$ and $v^- \in \tV^-$
derived from each vertex $v \in V$ must be contained in a single strongly connected component.

Hence, Problem (DMI) with the input bipartite graph balanced and perfectly matchable
is equivalent to making a directed graph strongly connected by adding a minimum number of edges,
which was solved by Eswaran and Tarjan \cite{ET1976} (cf. Theorem~\ref{thm:ET}).
Note that every strongly connected component of the auxiliary directed graph
intersects both $V^+$ and $V^-$ in this case,
and one can choose, freely in each strongly connected component,
the heads and tails of additional edges in the strong connectivity augmentation.
This equivalence is utilized in our algorithm
for the balanced case presented in Section~\ref{sec:Algorithm}.

\section{Reduction to Supermodular Covering}\label{sec:FJ}
In this section, we show that Problem (DMI) is a special case of 
supermodular covering introduced by Frank and Jord\'an \cite{FJ1995}.
We first describe necessary definitions and the min-max duality theorem on supermodular covering
in Section~\ref{sec:Covering_bisupermodular}.
Next, in Section~\ref{sec:unbalanced_to_balanced},
we show a reduction of the unbalanced case of Problem (DMI) to the balanced case.
In Section \ref{sec:Formulation},
we then formulate the balanced case in terms of the Frank--Jord\'an framework.
Finally, via the reduction to supermodular covering,
we give a proof of our min-max duality theorem 
(Theorem~\ref{thm:duality}) in Section~\ref{sec:Proof_via_FJ}.

\subsection{Supermodular covering problem and min-max duality}\label{sec:Covering_bisupermodular}
Let $V^+$ and $V^-$ be finite sets.
Two ordered pairs $(X^+, X^-), (Y^+, Y^-) \in 2^{V^+} \times 2^{V^-}$
are said to be {\em dependent} if both $X^+ \cap Y^+$ and $X^- \cap Y^-$ are nonempty,
and {\em independent} otherwise.
A family $\cF \subseteq 2^{V^+} \times 2^{V^-}$ is called {\em crossing}
if, for every pair of dependent members $(X^+, X^-), (Y^+, Y^-) \in \cF$,
both $(X^+ \cap Y^+, X^- \cup Y^-)$ and $(X^+ \cup Y^+, X^- \cap Y^-)$ are also in $\cF$.

A function $g \colon \cF \to \ZZ_{\geq 0}$ on a crossing family
$\cF \subseteq 2^{V^+} \times 2^{V^-}$ is said to be
{\em crossing supermodular} if, for every pair of dependent members $(X^+, X^-), (Y^+, Y^-) \in \cF$
with $g(X^+, X^-) > 0$ and $g(Y^+, Y^-) > 0$, we have
\[g(X^+ \cap Y^+, X^- \cup Y^-) + g(X^+ \cup Y^+, X^- \cap Y^-) \geq g(X^+, X^-) + g(Y^+, Y^-).\]
We say that a multiset $F$ of directed edges in $V^+ \times V^-$ {\em covers}
a crossing supermodular function $g \colon \cF \to \ZZ_{\geq 0}$
if $|F(X^+, X^-)| \geq g(X^+, X^-)$ holds for every $(X^+, X^-) \in \cF$,
where $F(X^+, X^-)$ denotes the multiset obtained by restricting $F$ into $X^+ \times X^-$.
\begin{description}
  \setlength{\itemsep}{0mm}
\item[\underline{Problem (FJ)}]

\item[Input:]\vspace{1mm}
  A crossing supermodular function $g \colon \cF \to \ZZ_{\geq 0}$
  on a crossing family $\cF \subseteq 2^{V^+} \times 2^{V^-}$.
\item[Goal:]
  Find a minimum-cardinality multiset $F$ of directed edges in $V^+ \times V^-$ such that $F$ covers $g$.
\end{description}\vspace{1mm}

Frank and Jord\'{a}n \cite{FJ1995} showed a min-max duality on this problem as follows.

\begin{theorem}[Frank--Jord\'{a}n~{\cite[Theorem 2.3]{FJ1995}}]\label{thm:FJ1995}
  The minimum cardinality of a multiset $F$ of directed edges in $V^+ \times V^-$
  such that $F$ covers a crossing supermodular function $g \colon \cF \to \ZZ_{\geq 0}$ is equal to
  the maximum value of
  \begin{equation}\label{eq:FJ1995}
    \eta(\cS) := \sum_{(X^+, X^-) \in \cS} g(X^+, X^-),\nonumber
  \end{equation}
  taken over all subfamilies $\cS \subseteq \cF$
  whose members are pairwise independent.
\end{theorem}

\subsection{Reduction of the unbalanced case to the balanced case}\label{sec:unbalanced_to_balanced}
As mentioned several times,
the unbalanced case of Problem (DMI) can be reduced to the balanced case.
To show such a reduction,
we give a useful rephrasement of DM-irreducibility.

\begin{lemma}\label{lem:DM-irreducibility}
  A bipartite graph $G = (V^+, V^-; E)$ with $|V^+| \leq |V^-|$ 
  and $|V^-| \geq 2$ is DM-irreducible if and only if
  $|\Gamma_G(X^+)| \geq |X^+| + 1$ for every nonempty $X^+ \subseteq V^+$ with $|X^+| < |V^-|$.
\end{lemma}

\begin{proof}
By the definition \eqref{eq:deficiency} of $f_G \colon 2^{V^+} \to \ZZ$,
the condition $|\Gamma_G(X^+)| \geq |X^+| + 1$ is equivalent to $f_G(X^+) \geq 1$.
By Conditions 1--3 in Theorem~\ref{thm:DM},
the DM-irreducibility of $G$ is equivalent to $V_\infty = V$ when $|V^+| < |V^-|$,
and to $V_1 = V$ when $|V^+| = |V^-|$.
In both cases, $X_0^+ = V_0^+ = \emptyset$ minimizes $f_G$, and $f_G(\emptyset) = 0$.

Suppose that $|V^+| < |V^-|$.
Then, $G$ is DM-irreducible if and only if
$X_0^+ = \emptyset$ is a unique minimizer of $f_G$;
equivalently, $f_G(X^+) \geq 1$ for every nonempty $X^+ \subseteq V^+$,
which satisfies $|X^+| \leq |V^+| < |V^-|$.

Suppose that $|V^+| = |V^-| \geq 2$.
Then, $G$ is DM-irreducible if and only if
$f_G$ has exactly two minimizers $X_0^+ = \emptyset$ and $X_1^+ = V_1^+ = V^+$;
equivalently, $f_G(V^+) = 0$ and
$f_G(X^+) \geq 1$ for every nonempty $X^+ \subsetneq V^+$, which satisfies $|X^+| < |V^+| = |V^-|$.
Note that the former condition is automatically satisfied
by the latter condition as follows.
For any nonempty $X^+ \subsetneq V^+$
with $|X^+| = |V^+| - 1$ (such $X^+$ exists because $|V^+| = |V^-| \geq 2$),
the latter condition implies
\[1 \leq f_G(X^+) = |\Gamma_G(X^+)| - |X^+| \leq |V^-| - |X^+| = 1.\]
We then have $V^- \supseteq \Gamma_G(V^+) \supseteq \Gamma_G(X^+) = V^-$,
and hence $\Gamma_G(V^+) = V^-$, which leads to $f_G(V^+) = |V^-| - |V^+| = 0$.
\end{proof}

The next lemma gives a reduction of the unbalanced case to the balanced case.
That is, making an unbalanced bipartite graph $G$ DM-irreducible by adding edges
is equivalent to making the corresponding balanced bipartite graph $G'$
defined in Lemma~\ref{lem:balanced} DM-irreducible by adding edges,
where the set of usable additional edges is not changed.

\begin{lemma}\label{lem:balanced}
  For a bipartite graph $G = (V^+, V^-; E)$ with $|V^+| < |V^-|$,
  define a balanced bipartite graph $G' = (V^+ \cup Z^+, V^-; E')$ as follows$:$
  let $Z^+$ be a set of new vertices with $|Z^+| = |V^-| - |V^+|$ and
  $E' := E \cup (Z^+ \times V^-)$.
  Then, $G$ is DM-irreducible if and only if so is $G'$.
\end{lemma}

\begin{proof}
When $|V^-| \leq 1$, both $G$ and $G'$ are DM-irreducible.
Assume $|V^-| \geq 2$ in what follows.

Consider the set functions $f_G \colon 2^{V^+} \to \ZZ$ and $f_{G'} \colon 2^{V^+ \cup Z^+} \to \ZZ$
defined in \eqref{eq:deficiency}.
By Lemma~\ref{lem:DM-irreducibility}, $G$ is DM-irreducible if and only if
$f_G(X^+) \geq 1$ for every nonempty $X^+ \subseteq V^+$,
and so is $G'$ if and only if
$f_{G'}(X^+) \geq 1$ for every nonempty $X^+ \subsetneq V^+ \cup Z^+$.
By the definition of $E'$, for every $X^+ \subseteq V^+ \cup Z^+$
with $X^+ \cap Z^+ \neq \emptyset$, we have $\Gamma_{G'}(X^+) = V^-$,
which implies $f_{G'}(X^+) = |V^-| - |X^+| = |V^+ \cup Z^+| - |X^+|$.
Hence,
$f_{G'}(X^+) \geq 1$ for every $X^+ \subsetneq V^+ \cup Z^+$ with $X^+ \cap Z^+ \neq \emptyset$.
Since $f_G(X^+) = f_{G'}(X^+)$ for every $X^+ \subseteq V^+$,
the above two conditions for the DM-irreducibility of $G$ and of $G'$ are equivalent.
\end{proof}

\subsection{Formulation of the balanced case as supermodular covering}\label{sec:Formulation}
We show that the balanced case of Problem (DMI) reduces to Problem (FJ).
Let $G = (V^+, V^-; E)$ be a bipartite graph with $|V^+| = |V^-| = n \geq 2$.
Define a family $\cF \subseteq 2^{V^+} \times 2^{V^-}$ and a function $g \colon \cF \to \ZZ_{\geq 0}$ by
\begin{align}
  \cF &:= \{\, (X^+, X^-) \mid \emptyset \neq X^+ \subseteq V^+,\ \emptyset \neq X^- \subseteq V^-,\ E(X^+, X^-) = \emptyset \,\},\nonumber\\[1mm]
  g(X^+, X^-) &:= \max\{0, |X^+| + |X^-| - n + 1\}. \label{eq:requirement}
\end{align}
Then, $\cF$ is crossing because $E(X^+ \cup Y^+, X^- \cap Y^-)$ and $E(X^+ \cap Y^+, X^- \cup Y^-)$
are included in $E(X^+, X^-) \cup E(Y^+, Y^-)$
for every $X^+, Y^+ \subseteq V^+$ and $X^-, Y^- \subseteq V^-$,
and $g$ is crossing supermodular because the second part in the maximum is modular.

\begin{claim}\label{cl:reduction}
  An edge set $F \subseteq V^+ \times V^-$ covers $g$
  if and only if $G + F$ is DM-irreducible.
\end{claim}

\begin{proof}
{[``Only if'' part]}~
Suppose that $F \subseteq V^+ \times V^-$ covers $g$.
By Lemma~\ref{lem:DM-irreducibility}, to see the DM-irreducibility of $G + F$,
it suffices to show that $|\Gamma_{G + F}(X^+)| \geq |X^+| + 1$
for every nonempty $X^+ \subsetneq V^+$.
Fix such $X^+$, and let $X^- := V^- \setminus \Gamma_{G + F}(X^+) \subseteq V^- \setminus \Gamma_G(X^+)$.
If $X^- = \emptyset$, then $\Gamma_{G + F}(X^+) = V^-$,
which implies $|\Gamma_{G + F}(X^+)| = |V^-| = |V^+| \geq |X^+| + 1$.
Otherwise, $\emptyset \neq X^- \subseteq V^- \setminus \Gamma_G(X^+)$,
and hence $(X^+, X^-) \in \cF$.
Since $F$ covers $g$ and $F(X^+, X^-) = \emptyset$,
we have $0 \geq g(X^+, X^-) = \max\{0, |X^+| + |X^-| - n + 1\}$.
This means $0 \geq |X^+| + |X^-| - n + 1 = |X^+| - |\Gamma_{G + F}(X^+)| + 1$,
and hence $|\Gamma_{G + F}(X^+)| \geq |X^+| + 1$.

\medskip
[``If'' part]~
Suppose that $G + F$ is DM-irreducible for $F \subseteq V^+ \times V^-$.
Then, by Lemma~\ref{lem:DM-irreducibility},
we have $|\Gamma_{G + F}(X^+)| \geq |X^+| + 1$
for every nonempty $X^+ \subsetneq V^+$.
For any $(X^+, X^-) \in \cF$, since $\Gamma_G(X^+) \cap X^- = \emptyset$,
we have $|F(X^+, X^-)| \geq |\Gamma_{G + F}(X^+) \cap X^-|$.
It is easy to see that $|\Gamma_{G + F}(X^+) \cap X^-| \geq |\Gamma_{G + F}(X^+)| - |V^- \setminus X^-| \geq |X^+| + 1 + |X^-| - n$,
which coincides with $g(X^+, X^-)$ when $g(X^+, X^-) > 0$.
Thus $F$ covers $g$.
\end{proof}

Since parallel edges make no effect on the DM-decomposition,
which is defined only by the adjacency relation
(cf.~the definition \eqref{eq:deficiency} of $f_G$),
the minimum of $|F|$ for covering a crossing supermodular function $g$
defined by \eqref{eq:requirement}
is attained by an edge ``set'' $F \subseteq V^+ \times V^-$.
Thus, Problem (DMI) reduces to Problem (FJ).
Since the values of $g$ are bounded by $n + 1$,
this problem is solved in polynomial time
by the pseudopolynomial-time algorithm of V\'{e}gh and Bencz\'{u}r \cite{VB2008}.

\subsection{Proof of the min-max duality (Theorem~\ref{thm:duality})}\label{sec:Proof_via_FJ}
Now we are ready to derive Theorem~\ref{thm:duality} from Theorem~\ref{thm:FJ1995}.
We postpone to Appendix \ref{sec:app2} the proof of Theorem~\ref{thm:duality_unbalanced}
via the reduction to the balanced case,
and prove it via matroid intersection instead in Section \ref{sec:WMI}.

We show $\max_\cX \tau_G(\cX) = \max_\cS \eta(\cS)$,
where the maxima are taken over all proper subpartitions $\cX$ of $V^+$ and of $V^-$
and all pairwise-independent subfamilies $\cS \subseteq \cF$.
We first confirm $\max_\cX \tau_G(\cX) \leq \max_\cS \eta(\cS)$.

\begin{claim}
  For any proper subpartition $\cX$\! of $V^+$\! or of $V^-$\!,
  there exists a pairwise-independent subfamily $\cS$ of $\cF$ such that
  $\tau_G(\cX) \leq \eta(\cS)$.
\end{claim}

\begin{proof}
By the symmetry, we assume that $\cX$ is a proper subpartition of $V^+$,
and define $X^- := V^- \setminus \Gamma_G(X^+)$ for each $X^+ \in \cX$.
Then, $(X^+, X^-) \in \cF$ $(X^+ \in \cX)$ are pairwise independent
(since $\cX$ is a subpartition of $V^+$),
and $g(X^+, X^-) = \max\{0, |X^+| - |\Gamma_G(X^+)| + 1\}$ by \eqref{eq:requirement}.
Hence, for $\cS := \{\, (X^+, X^-) \mid X^+ \in \cX \,\}$, we have
\[\tau_G(\cX) = \sum_{X^+ \in \cX} \left(|X^+| - |\Gamma_G(X^+)| + 1\right) \leq \sum_{X^+ \in \cX} g(X^+, X^-) = \eta(\cS).\qedhere\]
\end{proof}

In order to show the equality,
it suffices to show that, for any pairwise-independent subfamily $\cS \subseteq \cF$,
there exists a proper subpartition $\cY$ of $V^+$ or of $V^-$ such that
$\tau_G(\cY) \geq \eta(\cS)$.
Since any pair $(X^+, X^-) \in \cF$ with $g(X^+, X^-) = 0$
does not contribute to $\eta(\cS)$,
we assume that $g(X^+, X^-) > 0$ for every $(X^+, X^-) \in \cS$
by removing redundant pairs if necessary.
We then have $g(X^+, X^-) = |X^+| + |X^-| - n + 1 \leq |X^+| - |\Gamma_G(X^+)| + 1$ for every $(X^+, X^-) \in \cS$.
Let $\cS^\ast := \{\, X^\ast \mid (X^+, X^-) \in \cS \,\}$ for $\ast = +~\text{and}~-$.

\medskip
\noindent\underline{{\bf Case 1.}~~When $\cS^\ast$ is a subpartition of $V^\ast$ for $\ast = +~{\rm or}~-$.}

\medskip
By the symmetry, suppose that $\cS^+$ is a subpartition of $V^+$.
If $V^+ \not\in \cS^+$, then $\cY := \cS^+$ is a desired proper subpartition of $V^+$.
Otherwise, we have $\cS^+ = \{V^+\}$.
If $\cS^- \neq \{V^-\}$, then $\Gamma_G(X^-) = \emptyset$
and $g(V^+, X^-) = |X^-| + 1$ for a unique element $X^- \in \cS^-$,
and hence it suffices to take $\cY := \cS^-$.
Otherwise, $\cS = \{(V^+, V^-)\}$, and hence $E = E(V^+, V^-) = \emptyset$.
Note that $g(V^+, V^-) = n + 1$, and recall that we assume $n \geq 2$.
In this case, if we take a proper partition $\cY := \{\, \{u\} \mid u \in V^+ \,\}$ of $V^+$,
then
\[\tau_G(\cY) = \sum_{u \in V^+} \left(|\{u\}| - |\emptyset| + 1\right) = 2n \geq n + 1 = g(V^+, V^-) = \eta(\cS).\]

\noindent\underline{{\bf Case 2.}~~When $\cS^\ast$ is not a subpartition of $V^\ast$ for $\ast = +~{\rm and}~-$.}

\medskip
Since $X^+ \cap Y^+ = \emptyset$ or $X^- \cap Y^- = \emptyset$
for every distinct pairs $(X^+, X^-), (Y^+, Y^-) \in \cS$,
we have $|\cS| \geq 3$.
We shall show by induction on $|\cS|$ that
this case reduces to Case 1 by an {\em uncrossing} procedure.

We first observe that $V^+ \not\in \cS^+$ or $V^- \not\in \cS^-$.
Suppose to the contrary that $V^+ \in \cS^+$ and $V^- \in \cS^-$.
We then have $(V^+, X^-), (Y^+, V^-) \in \cS$
for some $X^- \subseteq V^-$ and $Y^+ \subseteq V^+$.
If $X^- = V^-$ or $Y^+ = V^+$,
then $(V^+, V^-) \in \cS$ cannot be independent from any other pair in $\cS \subseteq \cF$,
which contradicts $|\cS| \geq 3$.
Otherwise (i.e., if $X^- \neq V^-$ and $Y^+ \neq V^+$),
since $X^- \neq \emptyset \neq Y^+$ by the definition of $\cF$,
the two pairs $(V^+, X^-), (Y^+, V^-) \in \cF$ cannot be independent, a contradiction.
By the symmetry, we assume that $V^+ \not\in \cS^+$.

The following claim shows a successful uncrossing procedure.

\begin{claim}\label{cl:uncrossing}
  If distinct $X^+, Y^+ \in \cS^+$ satisfy
  $X^+ \cap Y^+ \neq \emptyset$ and $X^+ \cup Y^+ \neq V^+$,
  then one can reduce $|\cS|$ by
  replacing $(X^+, X^-)$ and $(Y^+, Y^-)$
  with $(X^+ \cap Y^+, X^- \cup Y^-)$ without reducing the value of $\eta(\cS)$.
\end{claim}

\begin{proof}
We first see that $(X^+ \cap Y^+, X^- \cup Y^-) \in \cF$.
This follows from $X^+ \cap Y^+ \neq \emptyset$ and
$E(X^+ \cap Y^+, X^- \cup Y^-) \subseteq E(X^+, X^-) \cup E(Y^+, Y^-) = \emptyset$.

Next, we confirm that $(X^+ \cap Y^+, X^- \cup Y^-)$ is independent from
each $(Z^+, Z^-) \in \cS \setminus \{(X^+, X^-), (Y^+, Y^-)\}$.
Since $(Z^+, Z^-)$ is independent from both $(X^+, X^-)$ and $(Y^+, Y^-)$,
at least one of $X^+ \cap Z^+$, $Y^+ \cap Z^+$, and $(X^- \cup Y^-) \cap Z^-$ is empty.
This implies that $(X^+ \cap Y^+) \cap Z^+ = \emptyset$ or $(X^- \cup Y^-) \cap Z^- = \emptyset$.

Finally, we show that the value of $\eta(\cS)$ does not decrease by this replacement.
Recall that $X^- \cap Y^- = \emptyset$ (since $X^+ \cap Y^+ \neq \emptyset$),
both $g(X^+, X^-)$ and $g(Y^+, Y^-)$ are positive, and $X^+ \cup Y^+ \subsetneq V^+$.
Thus we have the following inequalities, which complete the proof:
\begin{align*}
  &g(X^+ \cap Y^+, X^- \cup Y^-) \\[.5mm]
  &\quad \geq |X^+ \cap Y^+| + |X^- \cup Y^-| - n + 1\\[.5mm]
  &\quad = \left(|X^+| + |Y^+| - |X^+ \cup Y^+|\right) + \left(|X^-| + |Y^-|\right) - n + 1\\[.5mm]
  &\quad = \bigl(|X^+| + |X^-| - n + 1\bigr) + \bigl(|Y^+| + |Y^-| - n + 1\bigr) + \bigl(n - 1 - |X^+ \cup Y^+|\bigr)\\[.5mm]
  &\quad \geq g(X^+, X^-) + g(Y^+, Y^-). \qedhere
\end{align*}
\end{proof}

Some pair must be uncrossed by Claim~\ref{cl:uncrossing} as follows,
which completes the proof.

\begin{claim}\label{cl:existence}
  There exist distinct $X^+, Y^+ \in \cS^+$ such that
  $X^+ \cap Y^+ \neq \emptyset$ and $X^+ \cup Y^+ \neq V^+$.
\end{claim}

\begin{proof}
Suppose to the contrary that, for every distinct $X^+, Y^+ \in \cS^+$,
we have $X^+ \cap Y^+ = \emptyset$ or $X^+ \cup Y^+ = V^+$.
Take distinct elements $X^+, Y^+ \in \cS^+$ with $X^+ \cap Y^+ \neq \emptyset$,
and distinct pairs $(Z_1^+, Z_1^-), (Z_2^+, Z_2^-) \in \cS$ with $Z_1^- \cap Z_2^- \neq \emptyset$
(recall the case assumption that $\cS^*$ is not a subpartition of $V^\ast$ for $\ast = +~{\rm and}~-$).
Then, $X^+ \cup Y^+ = V^+$ and $Z_1^+ \cap Z_2^+ = \emptyset$.
We show that, for each $i \in \{1, 2\}$, exactly one of the following statements holds:
\begin{itemize}\vspace{-.5mm}
  \setlength{\itemsep}{.5mm}
\item[(a)]
  $Z_i^+ = X^+$;
\item[(b)]
  $Z_i^+ = Y^+$;
\item[(c)]
  $Z_i^+ \supseteq X^+ \triangle Y^+ := (X^+ \setminus Y^+) \cup (Y^+ \setminus X^+)$.
\end{itemize}\vspace{-.5mm}
Since $X^+ \setminus Y^+ \neq \emptyset \neq Y^+ \setminus X^+$
(otherwise, $X^+ = V^+$ or $Y^+ = V^+$, which contradicts that $V^+ \not\in \cS^+$),
every possible pair of (a)--(c) leads to $Z_1^+ \cap Z_2^+ \neq \emptyset$, a contradiction.

Suppose that $Z_i \neq X^+$ and $Z_i \neq Y^+$, and we derive Condition (c).
Since $X^+ \cup Y^+ = V^+$,
we assume $Z_i^+ \cap X^+ \neq \emptyset$ without loss of generality.
This implies $Z_i^+ \cup X^+ = V^+$, and hence
$Z_i^+ \supseteq V^+ \setminus X^+ = Y^+ \setminus X^+$.
Since $Y^+ \setminus X^+ \neq \emptyset$,
we also have $Z_i^+ \cap Y^+ \neq \emptyset$.
We then similarly see $Z_i^+ \supseteq X^+ \setminus Y^+$,
which concludes that $Z_i^+ \supseteq X^+ \triangle Y^+$.
\end{proof}

\section{Solving Unbalanced Case via Matroid Intersection}\label{sec:WMI}
In this section, we discuss a reduction of the unbalanced case of Problem (DMI)
to matroid intersection.
The readers are referred to \cite{Frank2011, Schrijver2003}
for basics on matroids and matroid intersection.

First, in Section~\ref{sec:minimal},
we introduce the concept of minimal DM-irreducibility
and give a simple characterization.
With the aid of the characterization,
we reduce the unbalanced case to matroid intersection in Section~\ref{sec:reduction_WMI}.
We also discuss the tractability of the minimum-cost augmentation problem in Section~\ref{sec:cost}.
In Section~\ref{sec:SCBG},
we show that our reduction can be derived also from a general framework of
covering supermodular functions by bipartite graphs.
Finally, in Section~\ref{sec:duality_MI},
we give a proof of the min-max duality (Theorem~\ref{thm:duality_unbalanced})
with the aid of Edmonds' matroid intersection theorem \cite{Edmonds1970}.

\subsection{Minimal DM-irreducibility}\label{sec:minimal}
We say that a subgraph $G'$ of a graph $G$ is {\bf spanning}
if $G'$ contains all the vertices in $G$ (some of which may be isolated), i.e.,
if $G'$ is obtained just by removing some edges from $G$.
Since the DM-irreducibility is not violated by adding edges,
a bipartite graph is DM-irreducible if and only if
it includes a minimal DM-irreducible spanning subgraph,
from which removing any edge violates the DM-irreducibility. 
We say that such a bipartite graph $G$ is {\em minimally DM-irreducible}, i.e.,
if $G$ itself is DM-irreducible but is no longer after removing an arbitrary edge.

To characterize the minimal DM-irreducibility, we use the following property of DM-irreducible graphs,
which immediately follows from the ``only if'' part of Lemma~\ref{lem:DM-irreducibility} with $X^+ = \{u\}$.

\begin{corollary}\label{cor:DM-irreducibility}
  If a bipartite graph $G = (V^+, V^-; E)$ with $|V^+| \leq |V^-|$ and $|V^-| \geq 2$ is DM-irreducible,
  then $|\Gamma_G(\{u\})| \geq 2$ for every $u \in V^+$.
\end{corollary}

The next lemma gives a simple characterization of
the minimally DM-irreducible unbalanced bipartite graphs,
which implies their matroidal structure.

\begin{lemma}\label{lem:minimally_DM-irreducible}
  A bipartite graph $G = (V^+, V^-; E)$ with $|V^+| < |V^-|$
  is minimally DM-irreducible if and only if $|\Gamma_G(\{u\})| = 2$ for every $u \in V^+$
  and $G$ is a forest as an undirected graph $($i.e., contains no undirected cycle$)$.
\end{lemma}

\begin{proof}
When $|V^+| = 0$, since $G = (\emptyset, V^-; \emptyset)$ is DM-irreducible,
the statement is trivial. Suppose that $|V^+| \geq 1$ and hence $|V^-| \geq 2$.

\medskip
{[``If'' part]}~
The DM-irreducibility follows from Claim~\ref{cl:1},
and the minimality is guaranteed by Corollary~\ref{cor:DM-irreducibility}.

\begin{claim}\label{cl:1}
  If $G$ is a forest such that $|\Gamma_G(\{u\})| \geq 2$ for every $u \in V^+$,
  then $G$ is DM-irreducible.
\end{claim}

\begin{proof}
Suppose to the contrary that
$G$ is a forest such that $|\Gamma_G(\{u\})| \geq 2$ for every $u \in V^+$
but $G$ is not DM-irreducible.
Then, by Lemma~\ref{lem:DM-irreducibility},
we have $|\Gamma_G(X^+)| \leq |X^+|$ for some nonempty $X^+ \subseteq V^+$.
Let $X^- := \Gamma_G(X^+)$ and $X := X^+ \cup X^-$.
Then, $G[X]$ contains $\sum_{u \in X^+} |\Gamma_G(\{u\})|$ edges
and $|X| = |X^+| + |X^-|$ vertices.
Since $\sum_{u \in X^+} |\Gamma_G(\{u\})| \geq 2|X^+| \geq |X^+| + |X^-|$,
there exists an undirected cycle in $G[X]$,
which is included in the forest $G$, a contradiction.
\end{proof}

[``Only if'' part]~
We first see that $G$ must be a forest.

\begin{claim}\label{cl:3}
  If $G$ is minimally DM-irreducible, then $G$ is a forest.
\end{claim}

\begin{proof}
By the DM-irreducibility,
$G$ has a perfect matching $M \subseteq E$,
and every vertex can reach some vertex in $V^- \setminus \partial^-M$ in $G(M) = G + \overline{M}$.
Let $H$ be the directed graph obtained from $G(M)$
by adding a new vertex $r$ and an edge $wr$ for each $w \in V^- \setminus \partial^-M$.
Then, every vertex is reachable to $r$ in $H$,
and hence $H$ contains a spanning {\em $r$-in-arborescence}
(a directed tree in which all edges are oriented toward $r$), say $T$,
which is obtained, e.g., by the depth first search from $r$
(where we traverse each edge in the backward direction).
Let $E_T \subseteq E$ be the set of edges
which or whose reverse edges appear in $T$.
Then, $E_T$ forms a forest that is also DM-irreducible,
and hence $E_T = E$ by the minimality.
\end{proof}

Combined with Corollary~\ref{cor:DM-irreducibility},
$G$ is a forest with $|\Gamma_G(\{u\})| \geq 2$ for every $u \in V^+$.
The equality in every inequality is guaranteed by Claim~\ref{cl:1} and the minimality.
\end{proof}

While Lemma \ref{lem:minimally_DM-irreducible}
provides a complete characterization of the minimal DM-irreducibility in the unbalanced case,
it is rather difficult to do so in the balanced case in the same manner.
One can, however, characterize at least the minimal DM-irreducibility
with the minimum number of edges as follows,
which is useful to show the NP-hardness of the minimum-cost augmentation
(see Section~\ref{sec:cost}).

\begin{lemma}\label{lem:minimally_DM-irreducible_balanced}
  Let $G = (V^+, V^-; E)$ be a bipartite graph with $|V^+| = |V^-| = n \geq 2$ and $|E| = 2n$.
  Then, $G$ is minimally DM-irreducible
  if and only if $G$ is connected and $|\Gamma_G(\{v\})| = 2$ for every $v \in V$,
  i.e., $G$ is isomorphic to a Hamiltonian cycle by ignoring the edge direction.
\end{lemma}

\begin{proof}
{[``If'' part]}~
Since $E$ can be partitioned into two disjoint perfect matchings,
$G$ is matching covered, which is equivalent to the DM-irreducibility under the connectivity.
The minimality immediately follows from Corollary~\ref{cor:DM-irreducibility}.

\medskip
[``Only if'' part]~
By the DM-irreducibility,
$G$ has a perfect matching $M \subseteq E$,
for which $G(M) = G + \overline{M}$ is strongly connected.
Hence, $G$ must be connected.
In addition, by Corollary~\ref{cor:DM-irreducibility},
we have $|\Gamma_G(\{v\})| \geq 2$ for every $v \in V$.
By the pigeonhole principle with $|E| = 2n = |V|$,
we conclude that $|\Gamma_G(\{v\})| = 2$ for every $v \in V$.
\end{proof}

\subsection{Reduction to matroid intersection}\label{sec:reduction_WMI}
We are now ready to reduce the unbalanced case to the matroid intersection problem.

First, Problem (DMI) is generally reformulated
as finding a minimum-weight minimally DM-irreducible spanning subgraph as follows.
For a given bipartite graph $G = (V^+, V^-; E)$,
define $\tE := V^+ \times V^-$, $\tG := (V^+, V^-; \tE)$,
and a weight function $\gamma \colon \tE \to \RR_{\geq 0}$ by
\begin{equation}
  \gamma(e) := \begin{cases}
    0 & (e \in E),\\
    1 & (e \in \tE \setminus E).
  \end{cases}\label{eq:w}
\end{equation}
For $\tF \subseteq \tE$,
we define its weight as $\gamma(\tF) := \sum_{e \in \tF} \gamma(e)$.
Then, making $G$ DM-irreducible by adding a smallest set $F \subseteq \tE \setminus E$
is equivalent to finding a minimum-weight edge set $\tF \subseteq \tE$
such that the spanning subgraph $(V^+, V^-; \tF)$ is minimally DM-irreducible
(recall that $G + F$ is DM-irreducible if and only if
$G + F$ includes a minimally DM-irreducible spanning subgraph).

Suppose that $\ell = |V^+| < |V^-| = n$.
Then, by Lemma~\ref{lem:minimally_DM-irreducible},
the set of minimally DM-irreducible spanning subgraphs of $\tG$
can be identified with the family of common independent sets of size $2|V^+| = 2\ell$
in the following two matroids on $\tE$:
\begin{itemize}\vspace{-.5mm}
  \setlength{\itemsep}{.5mm}
\item
  the cycle matroid $\tbM_1 = (\tE, \tcI_1)$ of $\tG$, i.e.,
  $\tF \in \tcI_1$ if and only if $\tF \subseteq \tE$ forms a forest;
\item
  a partition matroid $\tbM_2 = (\tE, \tcI_2)$ such that $\tF \in \tcI_2$
  if and only if at most two edges in $\tF \subseteq \tE$ leave each $u \in V^+$.
\end{itemize}\vspace{-.5mm}
Thus the unbalanced case reduces to finding a minimum-weight common independent set
of size $2\ell$ in the two matroids on $\tE$.

We show that this can be achieved by finding a maximum-cardinality common independent set
in the restrictions $\bM_i = (E, \cI_i)$ of $\tbM_i$ to $E \subseteq \tE$ for $i = 1, 2$,
which completes a reduction to matroid intersection.
The following claim gives a key observation.

\begin{claim}\label{cl:augment}
  For any $F \in \tcI_1 \cap \tcI_2$,
  there exists $\tF \in \tcI_1 \cap \tcI_2$ with $|\tF| = 2\ell$ and $F \subseteq \tF$.
  Moreover, such $\tF$ can be found in ${\rm O}(n)$ time.
\end{claim}

\begin{proof}
Since $F$ is a common independent set in $\tbM_1$ and $\tbM_2$,
the spanning subgraph $H := (V^+, V^-; F)$ of $\tG$ is a forest such that
$|\Gamma_{H}(\{u\})| \leq 2$ for every $u \in V^+$, and hence $|F| \leq 2\ell$.
It suffices to show that, when $|F| < 2\ell$,
there exists an edge $e \in \tE \setminus F$
such that $F \cup \{e\} \in \tcI_1 \cap \tcI_2$.

Suppose that $|F| < 2\ell$.
Then there exists a vertex $u \in V^+$ such that $|\Gamma_{H}(\{u\})| \leq 1$.
Let $H_u = (V_u^+, V_u^-; F_u)$ be the connected component of $H$ that contains $u$.
Since $H_u$ is a tree such that $|\Gamma_{H}(\{u\})| \leq 2$ for every $u' \in V_u^+ \setminus \{u\}$,
we have
\[|V_u^+| + |V_u^-| - 1 = |F_u| = \sum_{u' \in V_u^+} |\Gamma_{H}(u')| \leq 2|V_u^+| - 1,\]
which implies $|V_u^-| \leq |V_u^+| \leq |V^+| < |V^-|$.
Hence, there exists a vertex $w \in V^- \setminus V_u^-$,
for which the edge $e = uw \in \tE \setminus F$ can be added to $H$
so that the resulting spanning graph remains a forest with the degree constraint,
i.e, $F \cup \{e\} \in \tcI_1 \cap \tcI_2$.

One can add such edges $e \in \tE \setminus F$ simultaneously
by computing all the connected components of $H$ in advance,
which requires ${\rm O}(n)$ time in total.
\end{proof}

Let $\gamma^\ast := \min\{\, \gamma(\tF) \mid \tF \in \tcI_1 \cap \tcI_2 \text{ and } |\tF| = 2\ell \,\}$
and $q := 2\ell - \gamma^\ast$.

\begin{claim}\label{cl:q}
  The maximum cardinality of a common independent set in $\bM_1$ and $\bM_2$ is $q$.
\end{claim}

\begin{proof}
By the definition \eqref{eq:w} of the weight function $\gamma$,
for any $\tF \in \tcI_1 \cap \tcI_2$ with $|\tF| = 2\ell$ and $\gamma(\tF) = \gamma^\ast$,
the restriction $F := \tF \cap E \in \cI_1 \cap \cI_2$
satisfies $|F| = |\tF| - \gamma(\tF) = 2\ell - \gamma^\ast = q$.
To the contrary, by Claim~\ref{cl:augment},
for any $F \in \cI_1 \cap \cI_2 \subseteq \tcI_1 \cap \tcI_2$,
there exists $\tF \in \tcI_1 \cap \tcI_2$ with $|\tF| = 2\ell$ and $F \subseteq \tF$,
which implies $|F| \leq 2\ell - \gamma(\tF) \leq 2\ell - \gamma^\ast = q$.
\end{proof}

Finally, we confirm that a minimum-weight common independent set $\tF \in \tcI_1 \cap \tcI_2$
is obtained from a maximum-cardinality common independent set $F \in \cI_1 \cap \cI_2$, i.e., $|F| = q$.
By Claim~\ref{cl:augment},
one can find $\tF \in \tcI_1 \cap \tcI_2$ with $|\tF| = 2\ell$ and $F \subseteq \tF$,
which implies $\gamma(\tF) \leq 2\ell - |F| = 2\ell - q = \gamma^\ast$.
By the minimality of $\gamma^\ast$, indeed $\gamma(\tF) = \gamma^\ast$.

\medskip
In the resulting matroid intersection instance,
the ground set is of size $|E| = m$ and the optimal value (i.e., the maximum size of a common independent set) is at most $2|V^+| = {\rm O}(\ell)$.
With the aid of a fast ``graphic'' matroid intersection algorithm
due to Gabow and Xu \cite{GX1989, GX1996},
one can solve it in ${\rm O}(m\sqrt{\ell} \log \ell)$ time in general and
in ${\rm O}(m\sqrt{\ell})$ time when $m = \Omega(\ell^{1 + \epsilon})$ for some $\epsilon > 0$.

While $\bM_1$ is the cycle matroid of $G$ and hence is indeed graphic,
the other $\bM_2$, a partition matroid such that each upper bound is $2$,
is not graphic in general.
To use the graphic matroid intersection algorithm,
we duplicate the ground set $E$ by creating a copy $e' = uw$ of each element $e = uw \in E$,
and let $E'$ be the set of those copies.
Let $\bM_1' = (E \cup E', \cI_1')$ be the cycle matroid of the duplicated graph with the edge set $E \cup E'$,
in which each $e \in E$ and its copy $e' \in E'$ are parallel (i.e., $\{e, e'\} \not\in \cI_1'$).
Let $\bM_2' = (E \cup E', \cI_2')$ be the partition matroid
such that, for two subsets $F \subseteq E$ and $F' \subseteq E'$, we have $F \cup F' \in \cI_2'$
if and only if $F$ and $F'$ respectively have at most one edge leaving each $u \in V^+$.
Since each upper bound is $1$,
this $\bM_2'$ has a graphic representation as disconnected parallel edges according to the partition of $E \cup E'$.
The intersection of these two graphic matroids $\bM_1'$ and $\bM_2'$ is essentially the same as
the intersection of $\bM_1$ and $\bM_2$ by identifying each original element $e \in E$ and its copy $e' \in E'$,
where recall that $\{e, e'\} \not\in \cI_1'$.

\subsection{Minimum-cost augmentation}\label{sec:cost}
Our reduction technique can be utilized
even when, for each potential edge $e \in \tE \setminus E$,
the addition of $e$ gives rise to a cost of $c(e) \in \RR_{> 0}$
(note that, when $c(e) \leq 0$ for some $e$, we can add such $e$ to $G$ in advance).
We just need to modify the definition \eqref{eq:w}
of the weight function $\gamma \colon \tE \to \RR_{\geq 0}$ so that
$\gamma(e) = c(e)$ for each $e \in \tE \setminus E$.
Note that the original minimum-cardinality augmentation problem
is regarded as the case when $c(e) = 1$ for all $e \in \tE \setminus E$.
For this modified weight function $\gamma$,
we can no longer obtain a minimum-weight common independent set of size $2\ell$
by finding a maximum-cardinality common independent set in the restricted matroids,
but one can do in polynomial time by using weighted matroid intersection algorithms.

While we can reduce the ground set $\tE = V^+ \times V^-$ of two matroids to the original edge set $E$
in the minimum-cardinality augmentation case, we here need to use $\tE$ itself,
whose size $\tm := \ell n$ no longer depends on the number $m$ of original edges.
In general (when the cost values are arbitrary),
a weighted matroid intersection algorithm \cite{BCG1988}
for a partition matroid and a graphic matroid
leads to a bound on the computational time by
${\rm O}(\tm n + n^2 \ell + n\ell^2) = {\rm O}(n^2\ell)$.
Furthermore,
when the cost values are integers that is bounded by a constant,
weighted matroid intersection can be solved
by solving unweighted instances repeatedly 
in the same asymptotic running time bound \cite{HKK2016}.
Hence, by using the Gabow--Xu algorithm \cite{GX1989, GX1996}
for unweighted graphic matroid intersection,
one can obtain a better bound ${\rm O}(\tm\sqrt{\ell}) = {\rm O}(n\ell^{1.5})$,
where note that $\tm = \ell n = \Omega(\ell^2)$.

In contrast, the minimum-cost augmentation is NP-hard in the balanced case
(note that it was already shown in \cite{ET1976} for the strong connectivity augmentation,
which is equivalent to making a perfectly-matchable balanced bipartite graph DM-irreducible
as seen in Section~\ref{sec:Equivalence_SC}).
Consider testing whether a given bipartite graph $G_1 = (V^+, V^-; E_1)$ with $|V^+| = |V^-| = n \geq 2$
contains an undirected Hamiltonian cycle or not, which is NP-hard \cite{Krishnamoorthy1975}.
Define $G := (V^+, V^-; \emptyset)$, $\tE := V^+ \times V^-$, $E_2 := \tE \setminus E_1$,
and $c \colon \tE \to \RR_{> 0}$ by $c(e) := i$ for each $e \in E_i$ $(i \in \{1, 2\})$.
Then, by Lemma~\ref{lem:minimally_DM-irreducible_balanced},
$G_1$ contains an undirected Hamiltonian cycle
if and only if one can make $G$ DM-irreducible by adding edges with the total cost at most $2n$.

\subsection{Connection to supermodular covering by bipartite graphs}\label{sec:SCBG}
We can derive a matroid intersection formulation also from a general framework
of covering supermodular functions by bipartite graphs (cf. \cite[Section 13.4]{Frank2011}).

For a finite set $S$,
a set function $g \colon 2^S \to \ZZ_{\geq 0}$ is said to be {\em intersecting supermodular} if
\[g(X \cup Y) + g(X \cap Y) \geq g(X) + g(Y)\]
holds for every pair of subsets $X, Y \subseteq S$ with $X \cap Y \neq \emptyset$.
In addition, $g$ is {\em element-subadditive} if
\[g(X) + g(\{e\}) \geq g(X \cup \{e\})\]
holds for every pair of a subset $X \subseteq S$ and an element $e \in S \setminus X$.

Let $G = (V^+, V^-; E)$ be a bipartite graph.
We say that an edge set $F \subseteq E$ {\em covers}
a set function $g \colon 2^{V^+} \to \ZZ_{\geq 0}$ if
$|\Gamma_{F}(X^+)| \geq g(X^+)$ for every $X^+ \subseteq V^+$,
where we define $\Gamma_F(X^+) := \{\, w \mid \exists e = uw \in F \colon u \in X^+ \,\}$.
The following theorem gives a matroid intersection formulation
of covering an element-subadditive intersecting supermodular function by a bipartite graph.

\begin{theorem}[Frank {\cite[Theorem 13.4.11]{Frank2011}}]\label{thm:Frank}
  Let $G = (V^+, V^-; E)$ be a bipartite graph,
  and $g \colon 2^{V^+} \to \ZZ_{\geq 0}$ an element-subadditive intersecting supermodular function.
  If $E$ covers $g$,
  then all the minimal edge sets that cover $g$ form
  a family of all common independent sets of size $\sum_{u \in V^+} g(\{u\})$ in two matroids on $E$.
\end{theorem}

In order to apply Theorem \ref{thm:Frank} to our setting,
we define a set function $g \colon 2^{V^+} \to \ZZ_{\geq 0}$ by
\[g(X^+) := \begin{cases}
  0 & (X^+ = \emptyset),\\
  |X^+| + 1 & ({\rm otherwise}).
\end{cases}\]
As an easy observation,
this $g$ is indeed intersecting supermodular (the equality always holds) and element-subadditive
(since $g(\{u\}) = 2$ for every $u \in V^+$).
In addition, when $\ell = |V^+| < |V^-|$,
Lemma~\ref{lem:DM-irreducibility} implies that
an edge set $\tF \subseteq \tE = V^+ \times V^-$ covers $g$
if and only if the spanning subgraph $(V^+, V^-; \tF)$ of $\tG = (V^+, V^-; \tE)$ is DM-irreducible.
Hence, by Theorem~\ref{thm:Frank} (note that $\tE$ covers $g$),
all the minimally DM-irreducible spanning subgraphs of $\tG$
form a family of all common independent sets of size $2\ell$ in two matroids on $\tE$
(which indeed coincide with $\tbM_1$ and $\tbM_2$ defined in Section~\ref{sec:reduction_WMI}).

\subsection{Proof of the min-max duality (Theorem~\ref{thm:duality_unbalanced})}\label{sec:duality_MI}
In this section, we prove the min-max duality (Theorem~\ref{thm:duality_unbalanced})
through Edmonds' matroid intersection theorem \cite{Edmonds1970}.
We here adopt the definition of matroids by the rank functions.

\begin{theorem}[Edmonds {\cite[Theorem~(69)]{Edmonds1970}}]\label{thm:Edmonds}
  Let $\bM_1 = (E, \rho_1)$ and $\bM_2 = (E, \rho_2)$ be two matroids on the same ground set $E$.
  Then, the maximum cardinality of a common independent set in $\bM_1$ and $\bM_2$ is equal to
  the minimum value of
  \begin{equation}\label{eq:Edmonds}
    \rho_1(Z) + \rho_2(E \setminus Z),\nonumber
  \end{equation}
  taken over all subsets $Z \subseteq E$.
\end{theorem}

For a bipartite graph $G = (V^+, V^-; E)$ with $\ell = |V^+| < |V^-| = n$,
let $\bM_1 = (E, \rho_1)$ and $\bM_2 = (E, \rho_2)$ be the two matroids
defined in Section~\ref{sec:reduction_WMI},
i.e., $\bM_1$ is the cycle matroid of $G$ and $\bM_2$ is a partition matroid.
We denote by $q$ the maximum cardinality of a common independent set
in $\bM_1$ and $\bM_2$ (cf. Claim~\ref{cl:q} in Section~\ref{sec:reduction_WMI}).

We now start the proof of Theorem~\ref{thm:duality_unbalanced},
i.e., ${\rm opt}(G) = \max_{\cX^+} \tau_G(\cX^+)$,
where the maximum is taken over all subpartitions $\cX^+$ of $V^+$.
Since we have already seen ${\rm opt}(G) = \gamma^\ast = 2\ell - q$ in Section~\ref{sec:reduction_WMI}
and $q = \min_{Z \subseteq E} \left(\rho_1(Z) + \rho_2(E \setminus Z)\right)$
by Theorem~\ref{thm:Edmonds},
it suffices to confirm
\[\min_{Z \subseteq E} \left(\rho_1(Z) + \rho_2(E \setminus Z)\right) = 2\ell - \max_{\cX^+} \tau_G(\cX^+),\]
which is completed by Claims~\ref{cl:Edmonds1} and \ref{cl:Edmonds2}.

\begin{claim}\label{cl:Edmonds1}
  For any subpartition $\cX^+$ of $V^+$,
  there exists a subset $Z \subseteq E$ with
  \begin{equation}\nonumber
    \rho_1(Z) + \rho_2(E \setminus Z) \leq 2\ell - \tau_G(\cX^+).
  \end{equation}
\end{claim}

\begin{proof}
Let $\cX^+ = \{X_1^+, X_2^+, \ldots, X_k^+\}$ be a subpartition of $V^+$.
For each $i \in [k]$, define $X_i^- := \Gamma_G(X_i^+)$ and $X_i := X_i^+ \cup X_i^-$.
If $X_i^- \cap X_j^- \neq \emptyset$ for some distinct $i, j \in [k]$,
then replacing $X_i^+$ and $X_j^+$ with $X_i^+ \cup X_j^+$
does not decrease the value of $\tau_G(\cX^+)$ because
\[|\Gamma_G(X_i^+ \cup X_j^+)| = |X_i^- \cup X_j^-| = |X_i^-| + |X_j^-| - |X_i^- \cap X_j^-| \leq |X_i^-| + |X_j^-| - 1.\]
Hence, we can assume that $X_i \cap X_j = \emptyset$ for every distinct $i, j \in [k]$.

Let $Z \subseteq E$ be the set of edges induced by $X := \bigcup_{i \in [k]} X_i$.
Then, $\rho_1(Z) \leq \sum_{i = 1}^k \left(|X_i| - 1\right)$
and $\rho_2(E \setminus Z) \leq 2|X_0^+|$,
where $X_0^+ := V^+ \setminus X^+$. Thus we have
\begin{align*}
  \rho_1(Z) + \rho_2(E \setminus Z) \ &\leq \ \sum_{i = 1}^k \left(|X_i| - 1\right) + 2|X_0^+|\\
  &= \ \sum_{i = 1}^k \left(|X_i^+| + |X_i^-| - 1\right) + 2|X_0^+|\\
  &= \ 2\sum_{i = 0}^k |X_i^+| - \sum_{i = 1}^k \left(|X_i^+| - |X_i^-| + 1\right)\\
  &= \ 2\ell - \tau_G(\cX^+). \qedhere
\end{align*}
\end{proof}

\begin{claim}\label{cl:Edmonds2}
  For any subset $Z \subseteq E$,
  there exists a subpartition $\cX^+$ of $V^+$ with
  \begin{equation}\nonumber
    \rho_1(Z) + \rho_2(E \setminus Z) \geq 2\ell - \tau_G(\cX^+).
  \end{equation}
\end{claim}

\begin{proof}
For an edge set $Z \subseteq E$,
let $E_1 := Z$, $E_2 := E \setminus Z$,
and $H_i := (V^+, V^-; E_i)$ $(i = 1, 2)$.
We first show that we can assume the following two conditions:
\begin{itemize}\vspace{-.5mm}
  \setlength{\itemsep}{.5mm}
\item
  each vertex $u \in V^+$ is isolated in $H_1$ or in $H_2$;
\item
  if exactly one edge $e \in E$ leaves $u \in V^+$, then $e \in E_1$.
\end{itemize}

To see the first condition,
suppose to the contrary that, for some $u \in V^+$,
at least one edge leaves $u$ both in $H_1$ and in $H_2$.
Then, by transfering all the edges leaving $u$ in $H_1$ from $E_1$ to $E_2$,
the rank $\rho_1(E_1)$ decreases by at least 1 (since $u$ will be isolated in $H_1$)
and $\rho_2(E_2)$ increases by at most 1
(since $H_2$ already has at least one edge leaving $u$),
and hence the value of $\rho_1(Z) + \rho_2(E \setminus Z)$ does not increase.

To see the second condition,
suppose to the contrary that, for some $u \in V^+$,
exactly one edge $e \in E$ leaves $u \in V^+$ and $e \in E_2$.
Then, by transfering $e$ from $E_2$ to $E_1$,
the rank $\rho_1(E_1)$ increases by 1 (since $u$ is isolated in $H_1$)
and $\rho_2(E_2)$ decreases by 1 (since only $e$ leaves $u$ in $H_2$),
and hence the value of $\rho_1(Z) + \rho_2(E \setminus Z)$ does not change.

Let $Y^+ \subseteq V^+$ be the set of vertices that are not isolated in $H_2$,
and $\cX^+ := \{X_1^+, X_2^+, \ldots, X_k^+\}$ the partition of $X^+ := V^+ \setminus Y^+$
according to the connected components of $H_1 - Y^+ = G - Y^+$.
Then we have
\begin{align*}
  2\ell - \tau_G(\cX^+) \ &= \ 2|V^+| - \sum_{i = 1}^k \left(|X_i^+| - |\Gamma_G(X_i^+)| + 1\right)\\
  &= \ \sum_{i = 1}^k \left(|X_i^+| + |\Gamma_G(X_i^+)| - 1\right) + 2(|V^+| - |X^+|)\\
  &= \ \sum_{i = 1}^k \left(|X_i^+| + |\Gamma_{H_1}(X_i^+)| - 1\right) + 2|Y^+|\\
  &= \ \rho_1(Z) + \rho_2(E \setminus Z). \qedhere
\end{align*}
\end{proof}

\section{Algorithm for Balanced Case}\label{sec:Algorithm}
In this section, we present a direct algorithm for Problem (DMI) that only requires ${\rm O}(nm)$ time,
where the input bipartite graph $G = (V^+, V^-; E)$ is assumed to be balanced
with $|V^+| = |V^-| = n$ and $|E| = m$.
It should be remarked that our algorithm can solve the unbalanced case
through a reduction to the balanced case shown in Section~\ref{sec:unbalanced_to_balanced}
with the same computational time bound (see Appendix~\ref{sec:unbalanced}).

We describe our algorithm in Section~\ref{sec:Description}.
Next, in Section~\ref{sec:Optimality}, we show the optimality of the output,
which also gives an alternative, constructive proof of the min-max duality (Theorem~\ref{thm:duality}).
Finally, we analyze the running time of our algorithm in Section~\ref{sec:Time}.

\subsection{Algorithm description}\label{sec:Description}
We first compute the DM-decomposition of $G$,
say $(V_0; V_1, V_2, \ldots, V_k; V_\infty)$.
If $V_0 = V_\infty = \emptyset$, then $G$ has a perfect matching $M \subseteq E$.
In this case, it suffices to find a minimum number of additional edges
to make the auxiliary graph $G(M) = G + \overline{M}$ strongly connected
(as seen in Section~\ref{sec:Equivalence_SC}),
which can be done in linear time by Theorem~\ref{thm:ET}.

Otherwise, since $|V^+| = |V^-|$, both $V_0$ and $V_\infty$ are nonempty,
and hence $G$ has no perfect matching.
A possible strategy is to make $G$ perfectly matchable
by adding a perfect matching
$N \subseteq (V^+ \setminus \partial^+M) \times (V^- \setminus \partial^-M) 
\subseteq (V^+ \times V^-) \setminus E$ between the vertices exposed
by some maximum matching $M \subseteq E$ in $G$. 
The resulting graph $\tG := G + N$ has a perfect matching $\tM := M \cup N$,
and hence a minimum number of
further additional edges to make $\tG$ DM-irreducible can be found in linear time.
Thus we obtain a feasible solution, which may fail to be optimal.

We adopt a maximum matching $M \subseteq E$ in $G$
whose restrictions to $G[V_0]$ and to $G[V_\infty]$ 
are both {\em eligible} perfect matchings defined as follows.
This modification enables us to guarantee the optimality of the output
with the aid of the weak duality (Lemma~\ref{lem:weakdual}).

\begin{definition}\label{def:EPM}
  Let $H = (U^+, U^-; E)$ be a DM-irreducible unbalanced bipartite graph,
  and $M \subseteq E$ a perfect matching in $H$.
  When $|U^+| < |U^-|$, we say that $M$ is {\em eligible}
  if there exists a subpartition $\cX^-$ of $U^-$ such that
  $\tau_H(\cX^-) = |U^-| - |U^+| + s(H(M))$.
  Similarly, when $|U^+| > |U^-|$, we say so
  if there is a subpartition $\cX^+$ of $U^+$ such that
  $\tau_H(\cX^+) = |U^+| - |U^-| + t(H(M))$.
\end{definition}

Note that this definition is symmetric, i.e.,
the eligibility of $M$ when $|U^+| > |U^-|$ is
equivalent to the eligibility of $\overline{M}$
in the interchanged bipartite graph $(U^-, U^+; \overline{E})$.

Procedure EPM for finding an eligible perfect matching will be described in Section~\ref{sec:FEPM}.
A formal description of the entire algorithm is now given as follows.\clearpage
\begin{description}
  \setlength{\itemsep}{.5mm}
\item[\underline{Algorithm DMI$(G)$}]

\item[Input:]\vspace{.5mm}
  A bipartite graph $G = (V^+, V^-; E)$ with $|V^+| = |V^-| = n$.

\item[Output:]\vspace{-.5mm}
  An edge set $F \subseteq (V^+ \times V^-) \setminus E$
  with $|F| = {\rm opt}(G)$ such that $G + F$ is DM-irreducible.\vspace{1.5mm}

\item[Step 0.]
  Compute the DM-decomposition $(V_0; V_1, V_2, \ldots, V_k; V_\infty)$ of $G$.

\item[Step 1.]
  If $V_0 = V_\infty = \emptyset$,
  then set $N \leftarrow \emptyset$ and go to Step 4.

\item[Step 2.]
  Otherwise (i.e., if $V_0 \neq \emptyset \neq V_\infty$),
  find eligible perfect matchings $M_0 \subseteq E \cap (V_0^+ \times V_0^-)$ in $G[V_0]$
  and $M_\infty \subseteq E \cap (V_\infty^+ \times V_\infty^-)$ in $G[V_\infty]$
  by Procedure EPM.

\item[Step 3.]
  Take an arbitrary perfect matching 
  $N \subseteq (V_0^+ \setminus \partial^+M_0) \times (V_\infty^- \setminus \partial^-M_\infty)$.

\item[Step 4.]
  Let $\tG := G + N$, which has a perfect matching $\tM \subseteq E \cup N$.
  Using the Eswaran--Tarjan algorithm,
  find an edge set $\tF \subseteq (V^+ \times V^-) \setminus (E \cup N)$
  with $|\tF| = {\rm opt}(\tG)$
  such that $\tG(\tM) + \tF$ is strongly connected,
  and return $F \leftarrow N \cup \tF$.
\end{description}

\subsection{Optimality}\label{sec:Optimality}
In this section, we show that the output $F$ of Algorithm DMI$(G)$
is an optimal solution to Problem (DMI).
We first see the weak duality part of Theorem~\ref{thm:duality},
i.e., ${\rm opt}(G) \geq \max_\cX \tau_G(\cX)$.
We then construct a proper subpartition $\cX$ of $V^+$ or of $V^-$
such that $|F| = \tau_G(\cX)$, which implies that
$F$ and $\cX$ attain the minimum and the maximum, respectively.
The construction is presented separately for two cases:
when $G$ has a perfect matching and when not.
Note that the first case is not necessary for the optimality proof
(recall that it reduces to
the strong connectivity augmentation in Section~\ref{sec:Equivalence_SC}),
but is helpful to a discussion of the second case.

\subsubsection*{Weak duality}
\begin{lemma}\label{lem:weakdual}
  Let $G = (V^+, V^-; E)$ be a bipartite graph with $|V^+| = |V^-|$.
  Then, for any edge set $F \subseteq (V^+ \times V^-) \setminus E$
  such that $G + F$ is DM-irreducible
  and any proper subpartition $\cX$ of $V^+$ or of $V^-$, we have $|F| \geq \tau_G(\cX)$.
\end{lemma}

\begin{proof}
Fix an edge set $F \subseteq (V^+ \times V^-) \setminus E$ such that $G + F$ is DM-irreducible
and a proper subpartition $\cX$ of $V^+$.
By Lemma~\ref{lem:DM-irreducibility}, the DM-irreducibility of $G + F$
implies that $|\Gamma_{G + F}(X^+)| \geq |X^+| + 1$ for every $X^+ \in \cX$. Hence,
\[|F(X^+,\, V^- \setminus \Gamma_G(X^+))| \geq |\Gamma_{G + F}(X^+)| - |\Gamma_G(X^+)| \geq |X^+| - |\Gamma_G(X^+)| + 1,\]
where $F(Y^+, Y^-) := F \cap (Y^+ \times Y^-)$
denotes the restriction of $F$ to $Y^+ \times Y^-$
for $Y^+ \subseteq V^+$ and $Y^- \subseteq V^-$.
For every distinct $X_1^+, X_2^+ \in \cX$,
since $X_1^+ \cap X_2^+ = \emptyset$ implies
$F(X_1^+,\, V^- \setminus \Gamma_G(X_1^+)) \cap F(X_2^+,\, V^- \setminus \Gamma_G(X_2^+)) = \emptyset$, we see
\[|F| \geq \sum_{X^+ \in \cX} |F(X^+, V^- \setminus \Gamma_G(X^+))| \geq \sum_{X^+ \in \cX} \left(|X^+| - |\Gamma_G(X^+)| + 1\right) = \tau_G(\cX).\]

We can handle the proper subpartitions of $V^-$ in the same way
by considering the interchanged bipartite graph $(V^-, V^+; \overline{E})$
and the set $\overline{F}$ of reverse edges,
and thus we are done.
\end{proof}

\subsubsection*{Perfectly-matchable case}
Suppose that the input graph $G$ has a perfect matching $M \subseteq E$.
Then, Algorithm DMI$(G)$ just finds a minimum-cardinality set $F$ of additional edges
to make $G(M)$ strongly connected in Step 4.
If $G(M)$ itself is strongly connected,
then $\cX := \emptyset$ is a desired proper subpartition of $V^+$ (and of $V^-$),
i.e., $\tau_G(\cX) = 0 = |F|$.

Otherwise, $|F| = \max\{s(G(M)),\, t(G(M))\}$ by Theorem~\ref{thm:ET}.
Define two subpartitions $\cX^-$ of $V^-$ and $\cX^+$ of $V^+$
as follows (see also Fig.~\ref{fig:cX}):
\begin{align*}
  \cX^- &:= \{\, X^- \mid G(M)[X]~\text{is a source component of}~G(M) \,\},\\[1mm]
  \cX^+ &:= \{\, X^+ \mid G(M)[X]~\text{is a sink component of}~G(M) \,\},
\end{align*}
where recall that $X^+ := X \cap V^+$ and $X^- := X \cap V^-$ for $X \subseteq V$.
Since $G(M)$ is not strongly connected,
we have $\cX^- \neq \{V^-\}$ and $\cX^+ \neq \{V^+\}$.
We show that one of $\cX^-$ and $\cX^+$ is a desired proper subpartition
by confirming $\tau_G(\cX^-) = s(G(M))$ and $\tau_G(\cX^+) = t(G(M))$.

\begin{figure}[tb]
  \begin{center}
    \includegraphics[scale=0.5]{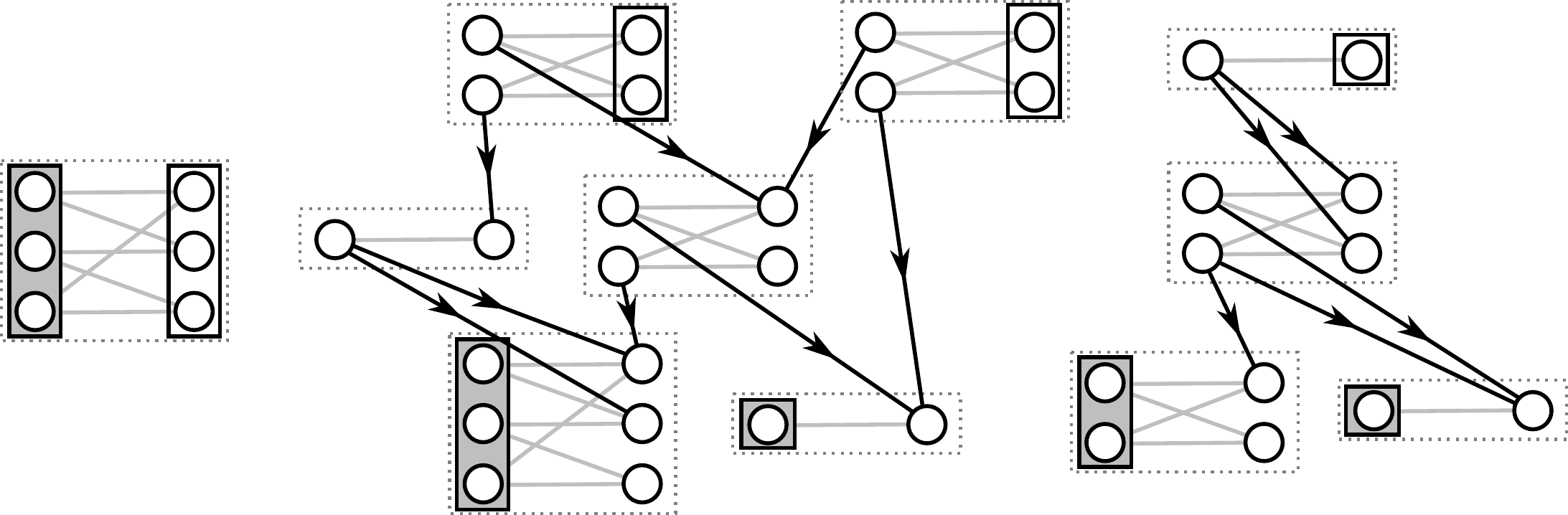}
    \caption{Proper subpartitions $\cX^-$ $($white boxes$)$ and $\cX^+$ $($gray boxes$)$ with $\tau_G(\cX^-) = s(G(M))$ and $\tau_G(\cX^+) = t(G(M))$ when $G$ has a perfect matching $M$, e.g., the set of all horizontal edges.}
    \label{fig:cX}
  \end{center}
\end{figure}

Since any edge in $M \cup \overline{M}$
is contained in some strongly connected component of $G(M)$,
distinct strongly connected components are connected only by edges
in $E \setminus M \subseteq V^+ \times V^-$.
Hence, for each source component $G(M)[X]$ of $G(M)$,
since no edge can enter $X$ in $G(M)$,
we have $\Gamma_G(X^-) = X^+$, which implies $|\Gamma_G(X^-)| = |X^+| = |X^-|$.
Similarly, for each sink component $G(M)[X]$ of $G(M)$, we have $|\Gamma_G(X^+)| = |X^-| = |X^+|$.
Thus we see
\[\tau_G(\cX^-) = \sum_{X^- \in \cX^-} 1 = |\cX^-| = s(G(M))~~\text{and}~~\tau_G(\cX^+) = \sum_{X^+ \in \cX^+} 1 = |\cX^+| = t(G(M)).\]

\subsubsection*{General case}
Suppose that the input graph $G$ has no perfect matching;
equivalently, $V_0 \neq \emptyset \neq V_\infty$ in the DM-decomposition
$(V_0; V_1, V_2, \ldots, V_k; V_\infty)$ of $G$.
In this case,
our algorithm finds a maximum matching $M \subseteq E$ in $G$ whose restrictions
$M_0$ to $G[V_0]$ and $M_\infty$ to $G[V_\infty]$
are both eligible perfect matchings in Steps 0 and 2
(cf. Condition 5 in Theorem~\ref{thm:DM} and the computation of the DM-decomposition in Section~\ref{sec:DM-decomposition}),
adds to $G$ a perfect matching
$N \subseteq (V^+ \setminus \partial^+M) \times (V^- \setminus \partial^-M)$
between the exposed vertices in Step 3 (see Fig.~\ref{fig:general}),
and finds an optimal solution $\tF \subseteq (V^+ \times V^-) \setminus (E \cup N)$
to $\tG = G + N$ in Step 4.

\begin{figure}[tb]
  \begin{center}
    \includegraphics[scale=0.7]{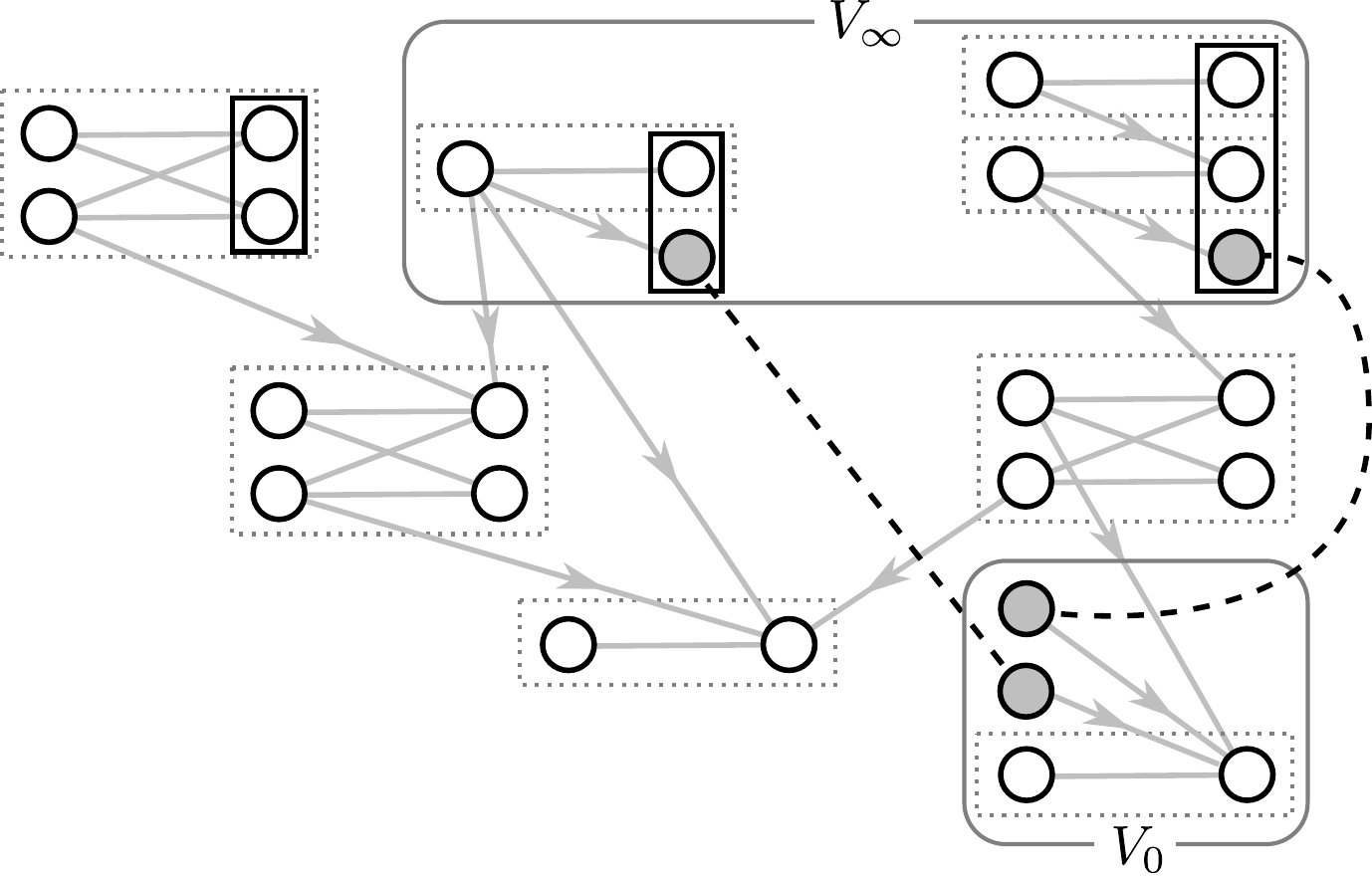}
    \caption{Illustration of the general case, where the maximum matching $M$ is the set of all horizontal edges, the perfect matching $N$ between the exposed vertices is drawn by dashed lines, and the white boxes represent the proper subpartition $\cX^- = \cX_\infty^- \cup \cX_\ast^-$ of $V^-$ with $\tau_G(\cX^-) = |V_\infty^-| - |V_\infty^+| + s(\tG(\tM))$.}
    \label{fig:general}
  \end{center}
\end{figure}

If $n = 1$, then $E = \emptyset$, $N = V^+ \times V^-$, and $\tF = \emptyset$.
Then the output $F = V^+ \times V^-$ is a unique feasible solution,
and hence optimal.
In what follows, we assume $n \geq 2$.
Then, as done above, 
it suffices to construct two proper subpartitions $\cX^-$ of $V^-$ and $\cX^+$ of $V^+$
such that $\max\{\tau_G(\cX^-),\, \tau_G(\cX^+)\} = |F| = |N| + |\tF|$.

Note that $|N| = n - |M| = |V_0^+| - |V_0^-| = |V_\infty^-| - |V_\infty^+|$.
The following claim implies $|\tF| = \max\{s(\tG(\tM)),\, t(\tG(\tM))\}$ by Theorem~\ref{thm:ET}, and hence
\begin{equation}\label{eq:|F|}
  |F| = \max\{|V_\infty^-| - |V_\infty^+| + s(\tG(\tM)),\, |V_0^+| - |V_0^-| + t(\tG(\tM))\},
\end{equation}
where $\tM := M \cup N$ is a perfect matching in $\tG$.

\begin{claim}\label{cl:SC}
  $\tG(\tM)$ is not strongly connected.
\end{claim}

\begin{proof}
By Observation~\ref{obs:source_sink},
each exposed vertex $u \in V^+ \setminus \partial^+M$
forms a source component of $G(M)$ which is reachable only to some vertices in $V_0$,
and each $w \in V^- \setminus \partial^-M$
forms a sink component of $G(M)$ which is reachable only from some vertices in $V_\infty$.
Since each edge $uw \in N$ connects such source and sink components one by one,
the two end vertices $u \in V^+$ and $w \in V^-$
form a new strongly connected component in $\tG(\tM) = G(M) + (N \cup \overline{N})$,
which is reachable only to some vertices in $V_0$ and only from some in $V_\infty$.
Recall that $|V^+| = |V^-| = n \geq 2$,
and hence $\tG(\tM)$ has at least two distinct strongly connected components.
\end{proof}

In what follows, we shall construct a subpartition $\cX^-$ of $V^-$
such that $\tau_G(\cX^-) = |V_\infty^-| - |V_\infty^+| + s(\tG(\tM))$ (see also Fig.~\ref{fig:general}).
By the symmetry, one can obtain a subpartition $\cX^+$ of $V^+$
such that $\tau_G(\cX^+) = |V_0^+| - |V_0^-| + t(\tG(\tM))$
in the same way (consider the interchanged bipartite graph $(V^-, V^+; \overline{E})$).
By \eqref{eq:|F|}, unless $\cX^- = \{V^-\}$ or $\cX^+ = \{V^+\}$,
these two subpartitions are desired ones.

Since no edge enters $V_\infty$ in $G$ as well as in $G(M)$ (see Observation~\ref{obs:source_sink})
and $M_\infty$ is an eligible perfect matching in $G_\infty := G[V_\infty]$,
there exists a subpartition $\cX_\infty^-$ of $V_\infty^-$ such that
$\tau_G(\cX_\infty^-) = \tau_{G_\infty}(\cX_\infty^-) = |V_\infty^-| - |V_\infty^+| + s(G_\infty(M_\infty))$.
Define
\[\cX_\ast^- := \{\, X^- \mid G(M)[X]~\text{is a source component of}~G(M)~\text{and}~X \cap (V_0 \cup V_\infty) = \emptyset \,\},\]
and $\cX^- := \cX_\infty^- \cup \cX_\ast^-$.
When $\cX^- \neq \{V^-\}$, the following claim completes the proof.

\begin{claim}\label{cl:cX^-}
  $\tau_G(\cX^-) = |V_\infty^-| - |V_\infty^+| + s(\tG(\tM))$.
\end{claim}

\begin{proof}
We first see $\tau_G(\cX^-) = |V_\infty^-| - |V_\infty^+| + s(G(M) - V_0)$.
Since no edge enters $V_\infty$ in $G(M)$,
the source components of $G(M) - V_0$ are partitioned into
those of $G(M)[V_\infty] = G_\infty(M_\infty)$ and those of $G(M)$ disjoint from $V_0 \cup V_\infty$.
Similarly to Section~\ref{sec:Equivalence_SC},
we see $\tau_G(\cX_\ast^-) = |\cX_\ast^-|$,
and hence $\tau_G(\cX^-) = \tau_G(\cX_\infty^-) + \tau_G(\cX_\ast^-) = |V_\infty^-| - |V_\infty^+| + s(G(M) - V_0)$.

Thus it suffices to show $s(\tG(\tM)) = s(G(M) - V_0)$.
Since no edge leaves $V_0$ in $G(M)$ and each source component of $G(M)[V_0]$
is a single exposed vertex $u \in V_0^+ \setminus \partial^+M$ with no entering edge,
the source components of $G(M)$ are partitioned into those of $G(M) - V_0$ and of $G(M)[V_0]$.
Hence, we have $s(G(M) - V_0) = s(G(M)) - s(G(M)[V_0])$.
Each exposed vertex $u \in V_0^+ \setminus \partial^+M$ is connected to some exposed vertex
$w \in V_\infty^- \setminus \partial^-M$ by two edges in $N \cup \overline{N}$.
As seen in the proof of Claim~\ref{cl:SC},
these two vertices $u$ and $w$ form a new strongly connected component in $\tG(\tM)$,
which is no longer a source component unless $w$ is isolated in $G(M)$,
i.e, the sink component $G_\infty(M_\infty)[\{w\}]$ is also a source component of $G(M) - V_0$.
Hence, whether some exposed vertices $w \in V_\infty^- \setminus \partial^-M$ are isolated or not,
by adding $N \cup \overline{N}$ to $G(M)$,
the number of source components decreases exactly by $s(G(M)[V_0])$.
Thus, $s(\tG(\tM)) = s(G(M)) - s(G(M)[V_0]) = s(G(M) - V_0)$.
\end{proof}

Finally, we consider the case of $\cX^- = \{V^-\}$.
Since $\tau(\cX_\infty^-) = |V_\infty^-| - |V_\infty^+| + s(G_\infty(M_\infty)) > 0$,
we have $\cX_\infty^- = \{V^-\}$ and $\cX_\ast^- = \emptyset$.
In this case, $V_\infty^- = V^-$ and $V_0^- = \emptyset$.
Hence, each vertex $u \in V_0^+$ is isolated in $G(M)$,
and is contained in a new sink component of $\tG(\tM) = G(M) + (N \cup \overline{N})$
consisting of two vertices.
Since the DM-decomposition of $G$ has no balanced component in this case,
we have $t(\tG(\tM)) = |V_0^+| = n - |M| \geq 1$, which leads to
\[|V_0^+| - |V_0^-| + t(\tG(\tM)) = 2(n - |M|) \geq n - |M| + 1 \geq |V^-| - |\Gamma_G(V^-)| + 1 = \tau_G(\cX_\infty^-).\]
Then the maximum in \eqref{eq:|F|} is attained by the latter term,
which is equal to $2|V_0^+|$.
Thus, for a subpartition $\cX^+ := \{\, \{u\} \mid u \in V_0^+ \,\} \neq \{V^+\}$ of $V^+$,
we have $\tau_G(\cX^+) = |F|$.

\subsection{Running time analysis}\label{sec:Time}
In this section, we show that Algorithm DMI$(G)$ runs in ${\rm O}(nm)$ time,
where recall that $n := |V^+| = |V^-|$ and $m := |E|$.

In Step 0, we find a maximum matching $M$ in $G$
and compute the strongly connected components of the auxiliary graph $G(M)$
(see Section~\ref{sec:DM-decomposition}).
The former can be done in ${\rm O}(nm)$ time even by a na\"{i}ve augmenting-path algorithm
(see, e.g., \cite[Section 16.3]{Schrijver2003}),
and the latter in ${\rm O}(n + m)$ time with the aid of the depth first search.
As shown in Section~\ref{sec:EPMR}, it takes ${\rm O}(nm)$ time to find
an eligible perfect matching, which is performed twice in Step 2.
Step 3 requires ${\rm O}(n)$ time,
and one can perform Step 4 in ${\rm O}(n + m)$ time by Theorem~\ref{thm:ET}
(note that a perfect matching $\tM$ in $\tG$
is obtained by combining the perfect matching $N \cup M_0 \cup M_\infty$ in $G[V_0 \cup V_\infty]$
with a perfect matching $M_\ast$ in $G - (V_0 \cup V_\infty)$,
which is included in the maximum matching $M$ in $G$ found in Step 0).
Thus the entire running time is bounded by ${\rm O}(nm)$.

\section{Finding Eligible Perfect Matchings}\label{sec:EPM}
In this section, we show a procedure
for finding an eligible perfect matching
in a DM-irreducible unbalanced bipartite graph $H = (U^+, U^-; E)$,
which plays a key role in Algorithm DMI.
Since the definition of eligibility is symmetric (see Definition~\ref{def:EPM}),
we assume $|U^+| < |U^-|$ in this section.

We describe an algorithm for finding an eligible perfect matching in 
Section~\ref{sec:FEPM}. Sections \ref{sec:EPMC} and \ref{sec:EPMR}
are devoted to its correctness proof and complexity analysis.

\subsection{Algorithm description}\label{sec:FEPM}
To describe the procedure,
we introduce an {\em augmented auxiliary graph}.

\begin{figure}[tb]
  \begin{center}
    \includegraphics[scale=0.5]{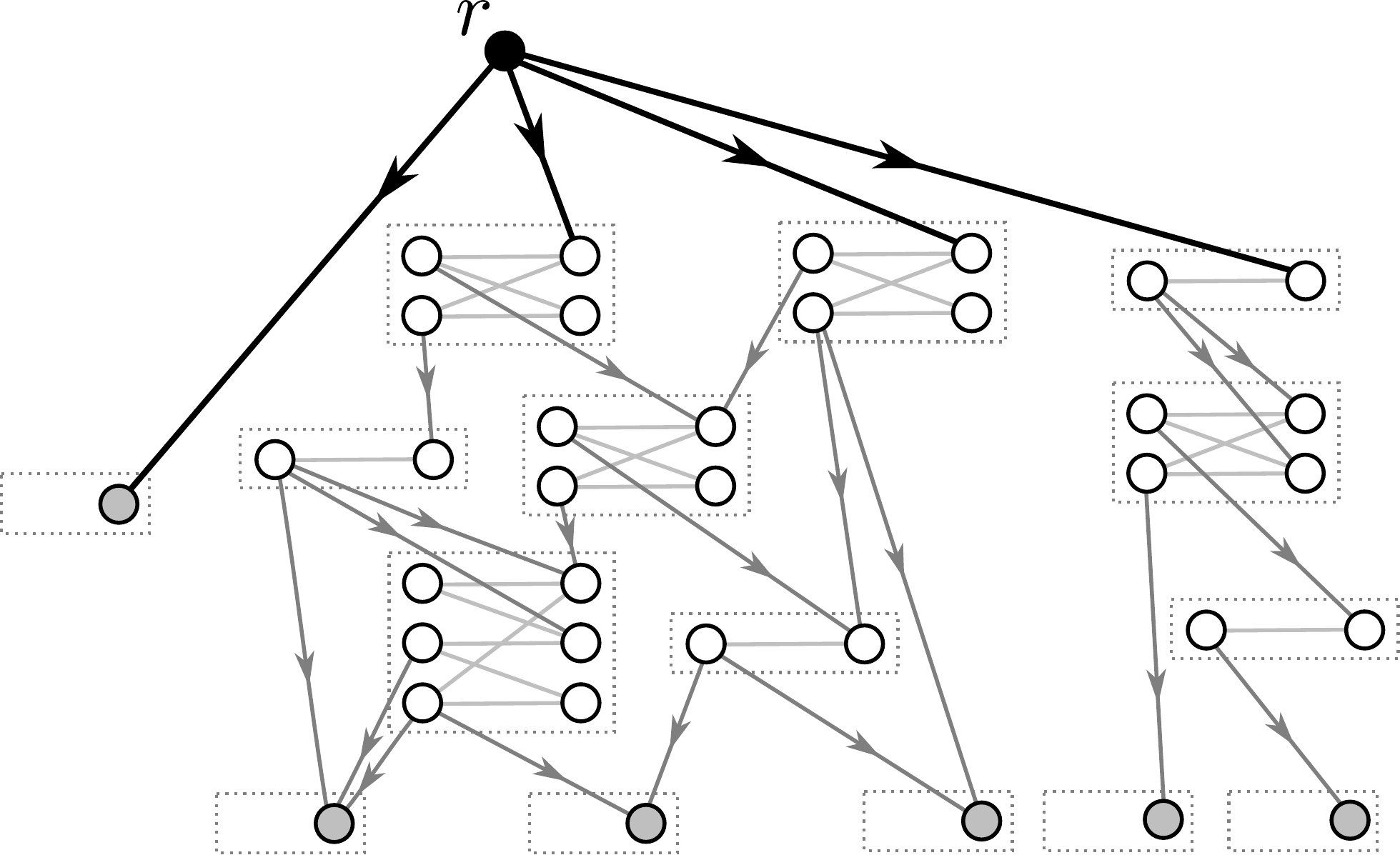}
    \caption{An augmented auxiliary graph $\hH(M)$, where the gray vertices are exposed by $M$.}
    \label{fig:H}
  \end{center}
\end{figure}

\begin{definition}\label{def:AAG}
  {\rm
  For a perfect matching $M \subseteq E$ in a DM-irreducible bipartite graph
  $H = (U^+, U^-; E)$ with $|U^+| < |U^-|$,
  an {\em augmented auxiliary graph} $\hH(M)$ is constructed
  from $H(M) = H + \overline{M}$ as follows (see also Fig.~\ref{fig:H}).
  Let $S^- \subseteq U^-$ be a vertex set obtained by collecting one vertex in $U^-$ 
  from each source component of $H(M)$, and hence, $|S^-| = s(H(M))$.
  Add to $H(M)$ a new vertex $r$ and an edge $rv$ for each $v \in S^-$.
  That is, $\hH(M) = (U \cup \{r\},\, E \cup \overline{M} \cup E_r)$,
  where $E_r := \{r\} \times S^-$.
  }
\end{definition}

Note that, since there may be several possible choices of $S^-$,
an augmented auxiliary graph $\hH(M)$ is not uniquely determined in general.

The procedure for finding an eligible perfect matching is now given as follows. 
\begin{description}
  \setlength{\itemsep}{.5mm}
\item[\underline{Procedure EPM$(H)$}]

\item[Input:]\vspace{.5mm}
  A DM-irreducible bipartite graph $H = (U^+, U^-; E)$ with $|U^+| < |U^-|$.

\item[Output:]\vspace{-.5mm}
  An eligible perfect matching $M \subseteq E$ in $H$.\vspace{1.5mm}

\item[Step 0.]
  Take an arbitrary perfect matching $M \subseteq E$ in $H$,
  and set $W \leftarrow U^- \setminus \partial^-M$.

\item[Step 1.]
  Construct an augmented auxiliary graph $\hH(M) = (U \cup \{r\},\, E \cup \overline{M} \cup E_r)$,
  and set $\hH = (\hU, \hE) \leftarrow \hH(M)$.

\item[Step 2.]
  While $W \neq \emptyset$, do the following.

  \setlength{\leftskip}{5mm}

\item[Step 2.1.]
  Take an exposed vertex $w \in W$, and update $W \leftarrow W \setminus \{w\}$.

\item[Step 2.2.]
  Find two edge-disjoint $r$--$w$ paths in $\hH$, or certify the nonexistence of such paths.

\item[Step 2.3.]
  If $\hH$ has two edge-disjoint $r$--$w$ paths,
  then let $P$ be one of those $r$--$w$ paths,
  and update $M \leftarrow (M \cup E(P)) \setminus M(P)$
  and $\hE \leftarrow (\hE \cup \overline{E(P)}) \setminus (\overline{M(P)} \cup \{e_1\})$
  (see Fig.~\ref{fig:flip}),
  where we denote by $E(P) \subseteq E$ the set of edges that appear in $P$,
  by $M(P) \subseteq M$ the set of edges whose reverse edges appear in $P$,
  and by $e_1 \in E_r$ the first edge of $P$.

  \setlength{\leftskip}{0mm}

\item[Step 3.]
  Return the current perfect matching $M$.
\end{description}\vspace{1mm}

\begin{figure}[tb]
  \begin{center}
    \includegraphics[scale=0.5]{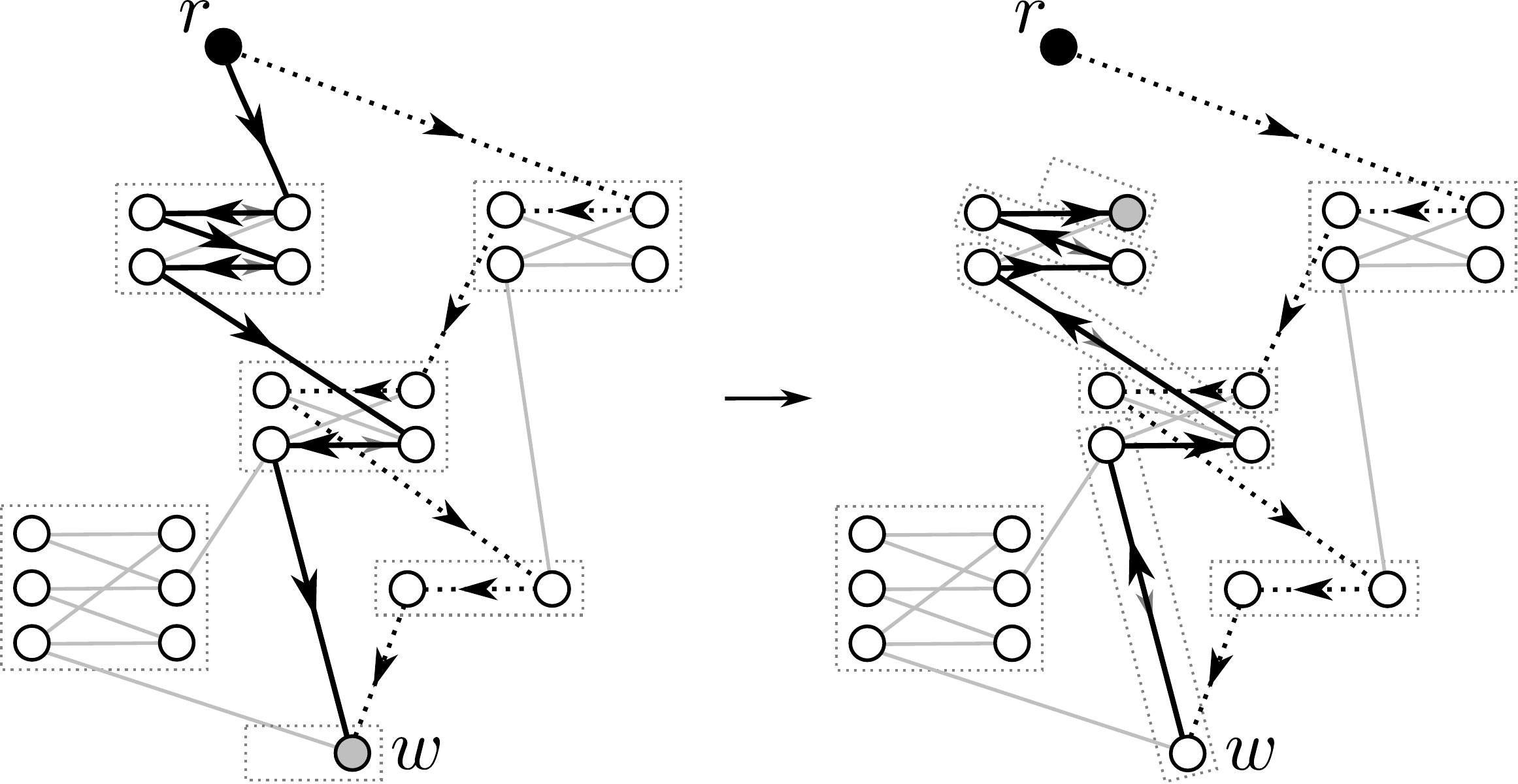}
    \caption{How $\hH$ is updated in Step $2.3$ of Procedure EPM along the bold $r$--$w$ path $P$.}
    \label{fig:flip}
  \end{center}
\end{figure}

The following lemma gives an important observation on Procedure EPM,
whose proof is left to Section~\ref{sec:EPMC}

\begin{lemma}\label{lem:key}
  At the beginning of each iteration of Step~$2$,
  $\hH = (\hU, \hE)$ is an augmented auxiliary graph $\hH(M)$,
  which does not have two edge-disjoint $r$--$w$ paths
  for any $w \in (U^- \setminus \partial^-M) \setminus W$.
\end{lemma}

\subsection{Correctness}\label{sec:EPMC}
We first give a proof of Lemma \ref{lem:key}, and then prove that Procedure EPM indeed outputs an eligible perfect matching.

\subsubsection*{Proof of Lemma~\ref{lem:key}}
We first see that $\hH$ is an augmented auxiliary graph with respect to $M$.

\begin{claim}\label{cl:H}
  After Step $1$, $\hH = (\hU, \hE)$ is always an augmented auxiliary graph $\hH(M)$.
\end{claim}

\begin{proof}
By Step 1, $\hH$ is initialized as $\hH(M)$.
We show that,
if the current perfect matching $M$ and an augmented auxiliary graph
$\hH = \hH(M) = (U \cup \{r\},\, E \cup \overline{M} \cup E_r)$
are updated to $M'$ and $\hH'$, respectively, in Step 2.3,
then $\hH'$ is an augmented auxiliary graph $\hH(M')$.

Let $v \in U^- \setminus \partial^-M'$ be the new exposed vertex,
and then $e_1 = rv \in E_r$.
Since $H(M') = H + \overline{M'}$ is obtained from $H(M) = H + \overline{M}$
by adding the edges in $\overline{E(P)}$ and removing those in $\overline{M(P)}$,
it suffices to show that the source components of $H(M')$ coincide with
those of $H(M)$ except for that containing $v$.

Let $X \subseteq U$ be the vertex set
of a source component of $H(M)$ with $v \not\in X$.
Then, since no edge enters $X$ in $\hH$ except for one in $E_r \setminus \{e_1\}$,
the $r$--$w$ path $P$ starting $e_1$ is disjoint from $X$.
Hence, $H(M')[X] = H(M)[X]$ remains a source component in $H(M')$ as it is in $H(M)$.

Suppose to the contrary that $H(M')$ has another source component $H(M')[Y]$.
If $P$ is disjoint from $Y$, then $H(M)[Y] = H(M')[Y]$ is a source component of $H(M)$,
and hence $v \in Y$, which however contradicts that $P$ is disjoint from $Y$.
Since $r \not\in Y$, the $r$--$w$ path $P$ must enter $Y$ at least once.
If $P$ leaves $Y$ using an edge $e \in E \cup \overline{M}$,
then the reverse edge $\bar{e}$ enters $Y$ in $H(M')$,
which contradicts that $H(M')[Y]$ is a source component.
Thus $P$ enters $Y$ exactly once, 
and $Y$ must contain the end $w$ of $P$.

Since $\hH$ has two edge-disjoint $r$--$w$ paths,
$Y$ has an entering edge $e$ in $\hH$ that does not appear in $P$.
If $e \in E \cup \overline{M}$, then $e$ remains in $H(M')$ as an edge entering $Y$,
a contradiction.
Otherwise, $e \in E_r \setminus \{e_1\}$.
This however contradicts that $Y$ is disjoint from
any source component of $H(M)$ that does not contain $v$. 
\end{proof}

When the procedure reaches Step 2 for the first time,
we have $W = U^- \setminus \partial^-M$,
and hence there is no choice of $w \in (U^- \setminus \partial^-M) \setminus W = \emptyset$.
We inductively show that, at the beginning of each iteration of Step 2,
$\hH$ does not have two edge-disjoint $r$--$w$ paths
for any $w \in (U^- \setminus \partial^-M) \setminus W$.
That is, we prove that,
if this property holds at the beginning of some iteration of Step 2,
then so does it at the end of the iteration
(equivalently, at the beginning of the next iteration).

Let $w^\ast \in W$ be the exposed vertex chosen in Step 2.1, and $W' := W \setminus \{w^\ast\}$.
If $\hH$ does not have two edge-disjoint $r$--$w^\ast$ paths,
then $M$ and $\hH$ are not updated.
In this case, combining with the induction hypothesis, we see that
$\hH$ does not have two edge-disjoint $r$--$w$ paths
for any $w \in ((U^- \setminus \partial^-M) \setminus W) \cup \{w^\ast\} = (U^- \setminus \partial^-M) \setminus W'$.

Suppose that $\hH = (\hU, \hE)$ has two edge-disjoint $r$--$w^\ast$ paths,
and $M$ and $\hH$ are updated to $M'$ and $\hH'$, respectively,
in Step 2.3. 
Let $v^\ast \in U^- \setminus \partial^-M'$ be the new exposed vertex, i.e., $e_1 = rv^\ast \in E_r$.
We then see $(U^- \setminus \partial^-M') \setminus W' = ((U^- \setminus \partial^-M) \setminus W) \cup \{v^\ast\}$,
and show that $\hH'$ does not have two edge-disjoint $r$--$w$ paths,
separately for $w = v^\ast$ and for $w \in (U^- \setminus \partial^-M) \setminus W$.

\begin{claim}
  $\hH'$ does not have two edge-disjoint $r$--$v^\ast$ paths.
\end{claim}

\begin{proof}
Since $v^\ast$ is in a source component of $H(M)$ that does not contain $w^\ast$,
its vertex set $X \subseteq U$ satisfies that
$v^\ast \in X$, $w^\ast \not\in X$, and $X$ has no entering edge in $H(M)$.
Hence, the $r$--$w^\ast$ path $P$ leaves $X$ exactly once through an edge $e \in E \cup \overline{M}$.
If $e \in \overline{M}$, then the reverse edge $\bar{e} \in M \subseteq E$ enters $X$ in $H(M)$,
a contradiction.
Otherwise, $e \in E$, which implies that $X$ has
a unique entering edge $\bar{e} \in \overline{M'}$ in $\hH'$.
Then, $v^\ast$ is not reachable from $r$ in $\hH' - \bar{e}$,
and hence $\hH'$ cannot have two edge-disjoint $r$--$v^\ast$ paths.
\end{proof}

In what follows, we show that $\hH'$ does not have two edge-disjoint $r$--$w$ paths
for any $w \in (U^- \setminus \partial^-M) \setminus W$.
Fix $w \in (U^- \setminus \partial^-M) \setminus W$.
Then, by the induction hypothesis and Menger's theorem \cite{Menger1927},
there exists an edge $e_w \in \hE$ such that $w$ is not reachable from $r$ in $\hH - e_w$.
One can choose such an edge so that $e_w \in \hE \setminus E = \overline{M} \cup E_r$ as follows.

\begin{claim}\label{cl:e_w}
  Choose an edge $e_w \in \hE$ so that
  the set $Y_w$ of vertices that are not reachable from $r$ in $\hH - e_w$
  contains $w$ and is maximal.
  Then, $e_w \not\in E$.
\end{claim}

\begin{proof}
By the definition, only $e_w$ enters $Y_w$ in $\hH$.
Suppose to the contrary that $e_w = uv \in E$
for some $u \in U^+ \setminus Y_w^+$ and $v \in Y_w^-$.
Since $M$ is a perfect matching in $H$,
there exists an edge $e' = v'u \in \overline{M}$ as well as $uv' \in M$ for some $v' \in U^-$.
If $v' \neq v$, then $v' \in U^- \setminus Y_w^-$.
Since only $e'$ enters $u \in U^+$ in $\hH$,
we can expand $Y_w$ to $Y_w \cup \{u\}$ by rechoosing $e_w$ as $e'$,
which contradicts the maximality of $Y_w$.
Otherwise, $e_w = uv \in M$.
Since only $e' = \bar{e}_w$ enters $u \in U^+$ in $\hH$,
every $r$--$u$ path in $\hH$ must intersect $v$,
and hence any $r$--$v$ path $Q$ in $\hH$ cannot traverse $e_w$.
Such a path $Q$ exists (since every vertex is reachable from $r$ in $\hH$
by the definition of an augmented auxiliary graph)
and enters $Y_w$ through an edge different from $e_w$ in $\hH$,
a contradiction.
\end{proof}

If $P$ is disjoint from $Y_w$,
then $w \in Y_w$ is not reachable from $r$ also in $\hH' - e_w$,
and hence $\hH'$ cannot have two edge-disjoint $r$--$w$ paths.
Otherwise, $P$ enters $Y_w$ through the edge $e_w \in \overline{M} \cup E_r$,
and leaves $Y_w$ at most once through an edge $e$.
Then, $e_w$ is no longer in $\hH'$,
and $Y_w$ has at most one new entering edge $\bar{e}$.
This also concludes that $\hH'$ cannot have two edge-disjoint $r$--$w$ paths.

\subsubsection*{Eligibility of output}
We here show that the output of Procedure EPM$(H)$ is indeed an eligible perfect matching. 
Suppose that EPM$(H)$ returns a perfect matching $M \subseteq E$ in $H$,
and let $\hH = (\hU, \hE)$ be the augmented auxiliary graph $\hH(M)$ when EPM$(H)$ halts,
where $\hU = U \cup \{r\}$ and $\hE = E \cup \overline{M} \cup E_r$.
Then, by Lemma~\ref{lem:key} and Menger's theorem \cite{Menger1927},
for any $w \in U^- \setminus \partial^-M$,
there exists an edge $e_w \in \hE$ such that $w$ is not reachable from $r$ in $\hH - e_w$.
Choose such an edge $e_w$ as in Claim~\ref{cl:e_w},
i.e., so that the set $Y_w$ of vertices that are not reachable from $r$ in $\hH - e_w$ is maximal.
We then see the following property.

\begin{claim}\label{cl:maximal}
  For any exposed vertices $w_1, w_2 \in U^- \setminus \partial^-M$,
  either $Y_{w_1} = Y_{w_2}$ or $Y_{w_1} \cap Y_{w_2} = \emptyset$.
\end{claim}

\begin{proof}
Let $w_1, w_2 \in U^- \setminus \partial^-M$ be distinct vertices,
and suppose to the contrary that $Y_{w_1} \neq Y_{w_2}$ and $Y_{w_1} \cap Y_{w_2} \neq \emptyset$.
We then have $e_{w_1} \neq e_{w_2}$.
If $Y_{w_1} \subsetneq Y_{w_2}$ or $Y_{w_2} \subsetneq Y_{w_1}$,
then we can expand the included one to the including one
by rechoosing $e_{w_1}$ or $e_{w_2}$ as the other one, respectively,
which contradicts the maximality of $Y_{w_1}$ and $Y_{w_2}$.
Thus, $Y_{w_1} \setminus Y_{w_2} \neq \emptyset \neq Y_{w_2} \setminus Y_{w_1}$.

Suppose that no edge enters $Y_{w_1} \cap Y_{w_2} \neq \emptyset$ in $\hH$.
Then, $\hH[Y_{w_1} \cap Y_{w_2}]$ has some source component of $\hH[U] = H(M)$,
which contradicts that $E_r$ contains an edge from $r \not\in Y_{w_1} \cap Y_{w_2}$
to each source component of $H(M)$.

Thus, $\hH$ has an edge $e$ entering $Y_{w_1} \cap Y_{w_2}$, which must be $e_{w_1}$ or $e_{w_2}$.
If $e$ enters $Y_{w_1} \cup Y_{w_2}$, then $e_{w_1} = e = e_{w_2}$, a contradiction.
Otherwise, assume that $e = e_{w_1}$ leaves $Y_{w_2} \setminus Y_{w_1}$ without loss of generality.
In this case, since $r \xrightarrow{\hH} w_2 \in Y_{w_1} \cup Y_{w_2}$,
the other edge $e_{w_2}$ must enter $Y_{w_1} \cup Y_{w_2}$.
This implies $Y_{w_2} \supseteq Y_{w_1} \cup Y_{w_2}$,
which contradicts $Y_{w_1} \setminus Y_{w_2} \neq \emptyset$.
\end{proof}

\begin{figure}[tb]
  \begin{center}
    \includegraphics[scale=0.6]{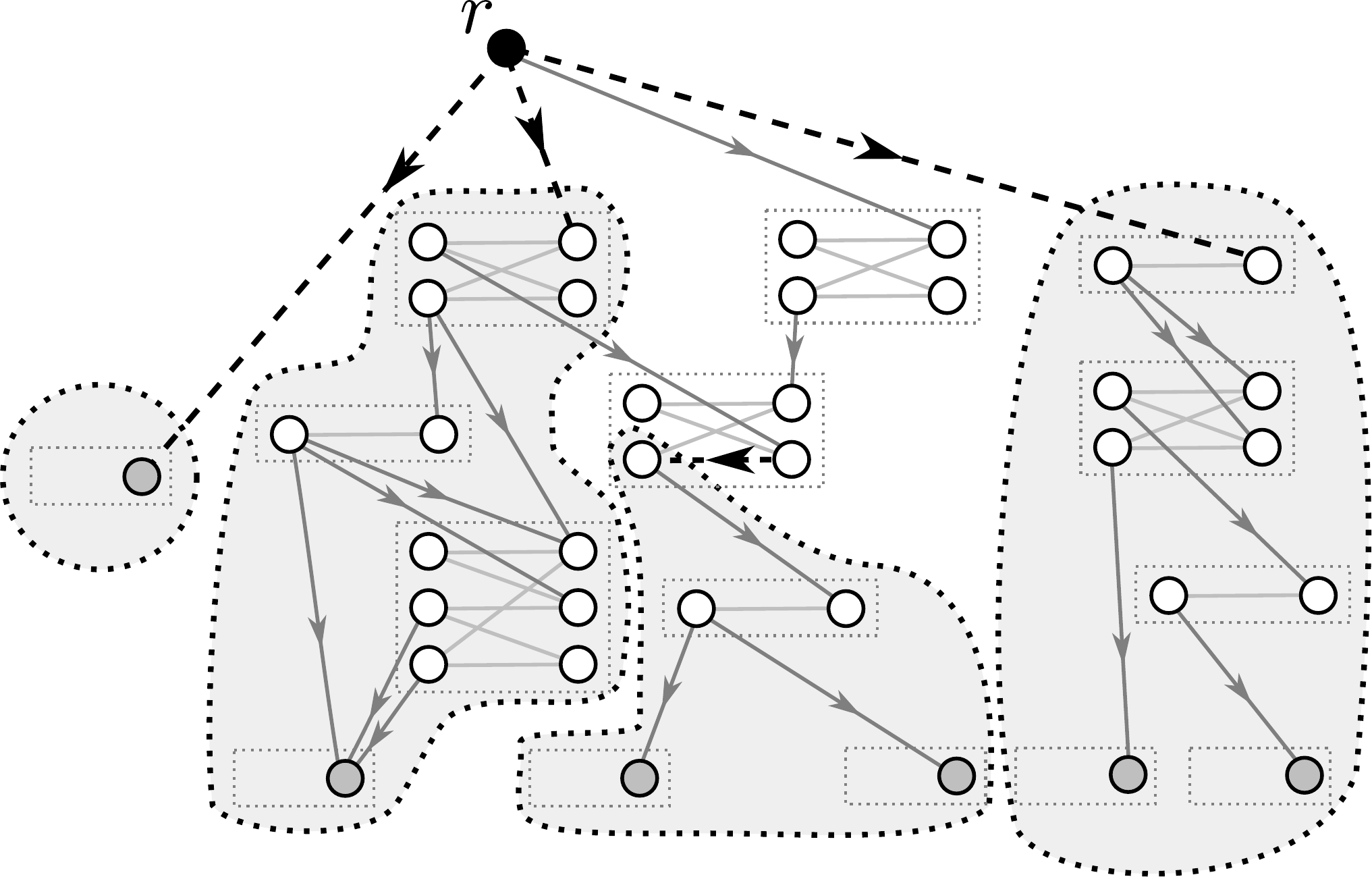}
    \caption{The subpartition of $U$ induced by $Y_w$ $(w \in U^- \setminus \partial^-M)$.}
    \label{fig:cY}
  \end{center}
\end{figure}

By Claim~\ref{cl:maximal},
$\{\, Y_w \mid w \in U^- \setminus \partial^-M \,\}$ is a subpartition of $U$
(see Fig.~\ref{fig:cY}).
Let $Y := \bigcup_{w \in U^- \setminus \partial^-M} Y_w$, and define
$\cX^- := \cY^- \cup \cZ^-$ as follows:
\begin{align*}
  \cY^- &:= \{\, Y_w^- \mid w \in U^- \setminus \partial^-M \,\},\\[1mm]
  \cZ^- &:= \{\, Z^- \mid H(M)[Z]~\text{is a source component of}~H(M)~\text{and}~Z \cap Y = \emptyset \,\},
\end{align*}
where recall that $X^+ := X \cap U^+$ and $X^- := X \cap U^-$ for $X \subseteq U$.
This $\cX^-$ is indeed a subpartition of $V^-$,
and we prove 
$\tau_H(\cX^-) = |U^-| - |U^+| + s(H(M))$.

By the definition \eqref{eq:duality}, 
we see $\tau_H(\cX^-) = \tau_H(\cY^-) + \tau_H(\cZ^-)$.
We first calculate $\tau_H(\cY^-)$
by evaluating $|Y_w^-| - |\Gamma_{H}(Y_w^-)| + 1$
for each exposed vertex $w \in U^- \setminus \partial^-M$.
Fix $w \in U^- \setminus \partial^-M$,
and let $T_w^- := Y_w \cap (U^- \setminus \partial^-M)$.
Then, by Claim~\ref{cl:maximal}, we have $Y_{w'} = Y_w$ for every $w' \in T_w^-$,
and $\{\, T_w^- \mid w \in U^- \setminus \partial^-M \,\}$
is a partition of $U^- \setminus \partial^-M$.
By Claim~\ref{cl:e_w},
we consider the following two cases separately:
when $e_w \in E_r$ and when $e_w \in \overline{M}$.

\begin{claim}\label{cl:E_r}
  If $e_w \in E_r$, then $|Y_w^-| - |\Gamma_{H}(Y_w^-)| + 1 = |T_w^-| + 1$.
\end{claim}

\begin{proof}
In this case, no edge enters $Y_w$ in $H(M)$.
Hence, each strongly connected component of $H(M)[Y_w]$ is also one of $H(M)$.
Since each sink component of $H(M)$
is a single vertex in $U^- \setminus \partial^-M$ (Observation~\ref{obs:source_sink})
and any other strongly connected component of $H(M)$ is balanced,
we see $|Y_w^-| = |Y_w^+| + |T_w^-|$.
Since only the edge $e_w \in E_r$ enters $Y_w$ in $\hH$
and every vertex in $Y_w$ is reachable in $\hH - e_w$ to some vertex in $T_w^- \subseteq Y_w^-$,
we see $\Gamma_H(Y_w^-) = Y_w^+$, and hence $|Y_w^-| - |\Gamma_{H}(Y_w^-)| + 1 = |T_w^-| + 1$.
\end{proof}

\begin{claim}\label{cl:M_0}
  If $e_w \in \overline{M}$,
  then $|Y_w^-| - |\Gamma_{H}(Y_w^-)| + 1 = |T_w^-|$.
\end{claim}

\begin{proof}
In this case,
$e_w = vu \in \overline{M}$ for some $v \in U^- \setminus Y_w^-$ and $u \in Y_w^+$.
Since only the edge $e_w$ enters $Y_w$ in $\hH$ and $M$ is a perfect matching in $H$,
any $u' \in Y_w^+ \setminus \{u\} \subseteq U^+$ is matched
with some $v' \in Y_w^- \setminus T_w^-$ by $M$, and vice versa.
Hence, $|Y_w^-| = |Y_w^+| - 1 + |T_w^-|$.
We observe $\Gamma_{H}(Y_w^-) = Y_w^+$ in the same way as the previous proof, and hence
$|Y_w^-| - |\Gamma_{H}(Y_w^-)| + 1 = |T_w^-|$.
\end{proof}

Let $\alpha := \bigl|\{\, Y_w \mid w \in U^- \setminus \partial^-M~\text{with}~e_w \in E_r \,\}\bigr|$.
By Claims~\ref{cl:E_r} and \ref{cl:M_0}, we see
\[\tau_H(\cY^-) = \sum_{Y_w^- \in \cY^-} |T_w^-| + \alpha = |U^- \setminus \partial^-M| + \alpha = |U^-| - |U^+| + \alpha.\]
Since the corresponding source component $H(M)[Z]$
is balanced for each $Z^- \in \cZ^-$
(which is disjoint from $Y \supseteq U^- \setminus \partial^-M$),
we see $\tau_H(\cZ^-) = |\cZ^-|$ (cf. Section~\ref{sec:Equivalence_SC}).
Hence, the next claim leads to $\alpha + \tau_H(\cZ^-) = s(H(M))$,
which completes the proof.

\begin{claim}\label{cl:alpha}
  $\alpha = \bigl|\{\, Z \mid H(M)[Z]~\text{\rm is a source component and}~Z \cap Y \neq \emptyset \,\}\bigr|$.
\end{claim}

\begin{proof}
We show that, for each $w \in U^- \setminus \partial^-M$,
exactly one source component of $H(M)$ intersects $Y_w$ if $e_w \in E_r$,
and so does no source component if $e_w \in \overline{M}$.
Since any strongly connected component of $H(M)[Y_w]$ is also one of $H(M)$ when $e_w \in E_r$,
a unique source component intersecting $Y_w$ is included in $H(M)[Y_w]$,
and hence this is sufficient for the claim.
Fix $w \in U^- \setminus \partial^-M$.

Suppose that $e_w = rv \in E_r$ for some $v \in S^- \cap Y_w^-$.
By the definition of $S^-$, the vertex $v$ is in a source component of $H(M)$.
Suppose to the contrary that there exists another source component of $H(M)$ intersecting $Y_w$.
Then, such a source component must be included in $H(M)[Y_w]$,
and hence there exists another edge $rv' \in E_r$ with $v' \in Y_w^-$.
This contradicts that only $e_w$ enters $Y_w^-$ in $\hH$.

Suppose that $e_w = vu \in \overline{M}$ for some $v \in U^- \setminus Y_w^-$ and $u \in Y_w^+$,
and to the contrary that there exists a source component $H(M)[Z]$ of $H(M)$
with $Z \cap Y_w \neq \emptyset$.
Then, by the definition of $S^-$,
there exists a vertex $v' \in S^- \cap Z$ with $e' = rv' \in E_r$.
If $v' \in Y_w$, then $e'$ enters $Y_w$ in $\hH$,
which contradicts that only $e_w \neq e'$ enters $Y_w$.
Otherwise, since $H(M)[Z]$ is strongly connected,
for any vertex $z \in Z \cap Y_w \neq \emptyset$,
there exists a $v'$--$z$ path in $H(M)[Z]$.
Such a path must traverse $e_w = vu$
(since only $e_w$ enters $Y_w$), and hence $\{u, v\} \subseteq Z$.
In this case, we can expand $Y_w$ to $Y_w \cup Z \supsetneq Y_w$
by rechoosing $e_w$ as $e'$, which contradicts the maximality of $Y_w$.
\end{proof}

\subsection{Running time analysis}\label{sec:EPMR}
In this section, we see that Procedure EPM$(H)$ runs in ${\rm O}(nm)$ time,
where $n := |U^-|$ and $m := |E|$ (note that $|U| = {\rm O}(n)$ since $|U^+| < |U^-|$).
Since the isolated vertices in $H$ can be ignored in the procedure
(which are added to $W$ in Step 0 and just discarded in Step 2.1),
we may assume $n = {\rm O}(m)$.

In Step 0, a perfect matching $M \subseteq E$ in $H$
can be found in ${\rm O}(nm)$ time even by a na\"{i}ve augmenting-path algorithm
(in fact, before calling this procedure,
one has been obtained in the course of computing the DM-decomposition).
In Step 1, since the strongly connected components of the auxiliary graph $H(M)$
are obtained in linear time,
an augmented auxiliary graph $\hH(M)$ is constructed in ${\rm O}(m)$ time.
Since $W$ is monotonically reduced in Step 2.1,
the number of iterations of Step 2 is $|W| = {\rm O}(n)$.
Step 2.2 can be done by performing the breadth first search twice
(i.e., by a na\"{i}ve augmenting-path algorithm originated by Ford and Fulkerson \cite{FF1956}),
which requires ${\rm O}(m)$ time.
The update of $M$ and $\hH$ along a path $P$ in Step 2.3 takes ${\rm O}(n)$ time.
Thus we conclude that the total computational time is bounded by ${\rm O}(nm)$.

\section{Applications}\label{sec:Applications}
We show possible applications of Problem (DMI)
in game theory and in control theory
(see \cite{CB2004} and \cite[Section 6.4]{Murota2000}, respectively, for the details).

\subsection{Bargaining in a two-sided market}
Consider bargaining in a two-sided market
with the seller set $S$ and the buyer set $B$
in which the tradable pairs are exogenously
given as a bipartite graph $G = (S, B; E)$,
where each edge in $E$ represents a tradable pair.
Each seller has an indivisible good and each buyer has money.
The bargaining process is repeated as described in the next paragraph,
and the utility received from a successful trade is defined as follows:
for a prescribed constant $\delta \in (0, 1)$,
if the trade is done at price $p$ at period $t \in \{0, 1, 2, \ldots \,\}$,
then the seller receives $\delta^tp$ and the buyer does $\delta^t(1 - p)$.
Note that all the sellers share one utility function, and so do all the buyers.

The bargaining process is as follows
(see \cite[Section 2.2]{CB2004} for the precise formulation).
All the sellers and all the buyers alternately offer prices in $[0, 1]$
for trade as the {\em proposers}.
Each agent in the other side accepts exactly one offered price
or rejects all of them as a {\em responder},
where the responders do not care with which specific proposer they trade.
For each price $p$ accepted by some responder,
restrict ourselves to the subgraph induced by the agents offering or accepting the price $p$,
and trade is done at price $p$ according to a maximum matching in the subgraph.
Note that there may be several possible choices of maximum matchings.
If there are multiple possibilities,
then one is chosen so that the set of matched agents
is lexicographically minimum in terms of the agent indices given in advance.
Note also that we are not concerned with
which specific edges are used in the maximum matching,
because the utility of each agent depends only on the price $p$ and the period $t$.
Remove all the agents who have traded from the graph,
and repeat the above process for the remaining graph until it has no edge.

A {\em subgame perfect equilibrium} in such a repeated game is, roughly speaking,
a strategy profile (i.e., in the above bargaining game,
the offering prices and the responses to offered prices
of all the agents at all the possible situations)
in which every agent has no incentive to change his or her action at any possible situation.
Corominas-Bosch \cite{CB2004} investigated the utility profile
in each subgame perfect equilibrium in the above game,
which is denoted by PEP for short (standing for a subgame Perfect Equilibrium Payoff).
She captured a typical utility profile
extending unique PEPs in several small markets,
called it the {\em reference solution},
and characterized when the reference solution is indeed a PEP
and moreover when it is a unique PEP.

\begin{theorem}[{Corominas-Bosch \cite[Theorem 1]{CB2004}}]\label{thm:CB2004_Thm1}
  Consider the above bargaining game on a bipartite graph $G = (S, B; E)$.\vspace{-.5mm}
  \begin{itemize}
    \setlength{\itemsep}{.5mm}
  \item
    When $G$ is unbalanced, the reference solution is a PEP
    if and only if $G$ is DM-irreducible.
  \item
    When $G$ is balanced, the reference solution is a PEP
    if and only if $G$ is perfectly matchable.
  \end{itemize}
\end{theorem}

\begin{theorem}[{Corominas-Bosch \cite[Proposition 6]{CB2004}}]\label{thm:CB2004_Prop6}
  Consider the above bargaining game on a bipartite graph $G = (S, B; E)$,
  and suppose that the game starts with the sellers' proposes.
  Then, the restriction of any PEP to $G_0$ is the reference solution to $G_0$,
  where $G_0 = (S_0, B_0; E_0)$ denotes the DM-irreducible component of $G$ with $|S_0| > |B_0|$.
  In particular, if $|S| > |B|$ and $G$ is DM-irreducible,
  then there exists a unique PEP, which is the reference solution.
\end{theorem}

Based on the above characterizations, for the unbalanced case,
our result gives a minimum number of additional tradable pairs
to make such a bargaining game admit a unique PEP,
which is the reference solution.
On the other hand, for the balanced case, the uniqueness of a PEP
is just guaranteed for the complete bipartite graphs \cite[Proposition 5]{CB2004}.
She also gave an example enjoying multiple PEPs,
in which the bipartite graph is not DM-irreducible.
What role the DM-decomposition of perfectly-matchable balanced bipartite graphs plays
in such bargaining has been left as an interesting question.

\subsection{Structural controllability of a linear system}
Consider a linear time-invariant system $(K,A,B)$ in a descriptor form 
$$K\dot{x}=Ax+Bu$$
with state variable $x$ and input variable $u$.
Under the genericity assumption that the set of nonzero entries in $K$, $A$, and $B$
are algebraically independent over $\QQ$, the system $(K, A, B)$
is said to be {\em structurally controllable} if the matrix pencil $A - sK$ is regular
(i.e., $\det(A - sK)\neq 0$ over the polynomial ring $\RR[s]$, where $s$ is an indeterminate)
and $[A - zK \mid B]$ is of row-full rank for every $z\in \CC$.

For a matrix pencil $D(s)$, let $G(D(s))$ denote the associated 
bipartite graph. The both-side vertex sets are the row set and the column set 
of $D(s)$, respectively, and the edges correspond to the nonzero entries of $D(s)$. 

\begin{theorem}[Murota {\cite[Corollary 6.4.8]{Murota2000}}]
  Let $(K,A,B)$ be a linear time-invariant system in a descriptor form 
  with nonsingular $K$. Under the genericity assumption, $(K,A,B)$ is 
  structurally controllable if and only if the following two conditions 
  hold.\vspace{-.5mm}
  \begin{itemize}
    \setlength{\itemsep}{.5mm}
  \item The bipartite graph $G([A\mid B])$ has a perfect matching.
  \item The bipartite graph $G([A-sK\mid B])$ is DM-irreducible.
  \end{itemize}
\end{theorem}

This characterization enables us to check efficiently 
if a given linear system is structurally controllable. 
If it turns out not to be, then a natural question is how 
to modify the system to make it structurally controllable. 
If $G([A\mid B])$ admits a perfect matching, our result
provides an answer to this question by identifying the minimum 
number of additional connections between the variables and the 
equations required to make the entire system structurally controllable. 

It would be more desirable if one can extend this approach to 
the case in which $G([A\mid B])$ may not have a perfect matching. 
It is also interesting to deal with the case of singular $K$. 
These problems are left for future investigation.

\section*{Acknowledgments}
We are grateful to L\'{a}szl\'{o}~A.~V\'{e}gh and Andr\'as Frank
for their insightful comments.
This work was supported
by the MTA-ELTE Egerv\'ary Research Group,
by the Hungarian National Research, Development and Innovation Office -- NKFIH grant K109240,
by JST CREST Grant Number JPMJCR14D2,
by JSPS KAKENHI Grant Number JP16H06931,
and by JST ACT-I Grant Number JPMJPR16UR.

\begin{appendix}
\section{On Reduction of Unbalanced Case to Balanced Case}\label{sec:app}
Although the unbalanced case is satisfactorily discussed via the reduction to matroid intersection in Section~\ref{sec:WMI},
we here provide an alternative discussion through the reduction
to the balanced case shown in Section \ref{sec:unbalanced_to_balanced}:
for an input unbalanced bipartite graph $G = (V^+, V^-; E)$ with $|V^+| < |V^-|$,
we construct a balanced bipartite graph $G' = (V^+ \cup Z^+, V^-; E')$
by adding a set $Z^+$ of new vertices that are adjacent to all the vertices in $V^-$,
i.e., $E' = E \cup (Z^+ \times V^-)$.

\subsection{Alternative proof of the min-max duality (Theorem~\ref{thm:duality_unbalanced})}\label{sec:app2}
In this section, we derive the min-max duality theorem for the unbalanced case
(Theorem~\ref{thm:duality_unbalanced}) from that for the balanced case (Theorem~\ref{thm:duality}).
First, we see the following weak duality as a corollary of Lemma~\ref{lem:weakdual}
(the weak duality in the balanced case) via the reduction. 

\begin{corollary}\label{cor:weakdual}
  Let $G = (V^+, V^-; E)$ be a bipartite graph with $|V^+| < |V^-|$.
  Then, for any edge set $F \subseteq (V^+ \times V^-) \setminus E$
  such that $G + F$ is DM-irreducible
  and any subpartition $\cX^+$ of $V^+$, we have $|F| \geq \tau_G(\cX^+)$.
\end{corollary}

We now start to prove Theorem~\ref{thm:duality_unbalanced}.
Let $G = (V^+, V^-; E)$ be a bipartite graph with $|V^+| < |V^-|$.
By Corollary~\ref{cor:weakdual},
it suffices to construct a subpartition $\cX^+$ of $V^+$ with $\tau_G(\cX^+) = {\rm opt}(G)$.
If $|V^-| = 1$, then $G$ itself is DM-irreducible,
and $\cX^+ := \emptyset$ is a subpartition of $V^+$ with $\tau_G(\cX^+) = 0 = {\rm opt}(G)$.
In what follows, we assume $|V^-| \geq 2$.

Let $G' = (V^+ \cup Z^+, V^-; E')$ be the balanced bipartite graph
that is constructed above. 
By Theorem~\ref{thm:duality},
there exists a proper subpartition $\cY$ of $V^+ \cup Z^+$ or of $V^-$
such that $\tau_{G'}(\cY) = {\rm opt}(G') = {\rm opt}(G)$.
Suppose that $\cY$ is a proper subpartition of $V^+ \cup Z^+$.
Since every vertex in $Z^+$ is adjacent to all the vertices in $V^-$,
for each $X^+ \subseteq V^+ \cup Z^+$ with $X^+ \cap Z^+ \neq \emptyset$,
we have $|X^+| - |\Gamma_{G'}(X^+)| + 1 = |X^+| - |V^-| + 1 \leq 0$.
By the maximality of $\tau_{G'}(\cY)$,
we may assume that $\cY$ contains no such $X^+$,
i.e., $\cY$ is a subpartition of $V^+$.
We then obtain a desired subpartition $\cX^+ := \cY$ of $V^+$
with $\tau_G(\cX^+) = \tau_{G'}(\cY) = {\rm opt}(G)$.

Otherwise, $\cY$ is a nonempty proper subpartition of $V^-$.
Suppose that $\cY$ contains two distinct elements $X^-, Y^- \in \cY$.
By the definition of $E'$,
we have $\emptyset \neq Z^+ \subseteq \Gamma_{G'}(X^-) \cap \Gamma_{G'}(Y^-)$,
which implies $|\Gamma_{G'}(X^- \cup Y^-)| = |\Gamma_{G'}(X^-) \cup \Gamma_{G'}(Y^-)| \leq |\Gamma_{G'}(X^-)| + |\Gamma_{G'}(Y^-)| - 1$.
Hence,
\[\left(|X^-| - |\Gamma_{G'}(X^-)| + 1\right) + \left(|Y^-| - |\Gamma_{G'}(Y^-)| + 1\right) \leq |X^- \cup Y^-| - |\Gamma_{G'}(X^- \cup Y^-)| + 1.\]
This enables us to replace $X^-$ and $Y^-$ with $X^- \cup Y^-$
without reducing the value of $\tau_{G'}(\cY)$.
Thus, by the maximality of $\tau_{G'}(\cY)$,
we may assume $\cY = \{Y^-\}$ for some nonempty $Y^- \subsetneq V^-$.
If $\Gamma_{G'}(Y^-) = V^+ \cup Z^+$,
then $\tau_{G'}(\cY) = |Y^-| - |V^+ \cup Z^+| + 1 = |Y^-| - |V^-| + 1 \leq 0$,
and hence $\cX^+ := \emptyset$ is a desired subpartition of $V^+$.
Otherwise, let $X^+ := V^+ \setminus \Gamma_G(Y^-) = (V^+ \cup Z^+) \setminus \Gamma_{G'}(Y^-) \neq \emptyset$ and $\cX^+ := \{X^+\}$.
We then see
\begin{align*}
  \tau_G(\cX^+) &= |X^+| - |\Gamma_{G}(X^+)| + 1\\[1mm]
  &= \left(|V^+ \cup Z^+| - |\Gamma_{G'}(Y^-)|\right) - |\Gamma_{G'}(X^+)| + 1\\[1mm]
  &= \left(|V^-| - |\Gamma_{G'}(X^+)|\right) - |\Gamma_{G'}(Y^-)| + 1\\[1mm]
  &\geq |Y^-| - |\Gamma_{G'}(Y^-)| + 1 = \tau_{G'}(\cY),
\end{align*}
which concludes that $\cX^+$ is a desired subpartition of $V^+$.

\subsection{Running time of Algorithm DMI}\label{sec:unbalanced}
The reduction to the balanced case increases the size of the input graph.
In particular, $G'$ may have an essentially larger number of edges than $G$,
i.e., $|E'| \neq {\rm O}(m)$, where $|V^+| < |V^-| = n$ and $|E| = m$.
While Algorithm DMI$(G')$ is just guaranteed
to run in ${\rm O}(n|E'|)$ time in Section~\ref{sec:Time},
it actually requires ${\rm O}(nm)$ time.
The following observation is useful to the analysis.

\begin{observation}\label{obs:balanced}
  Let $(V_0; V_1, V_2, \ldots, V_k; V_\infty)$ be the DM-decomposition of $G'$,
  and $M' \subseteq E'$ a maximum matching in $G'$.
  Then the following conditions hold.\vspace{-.5mm}
  \begin{itemize}
    \setlength{\itemsep}{.5mm}
  \item
    $M'$ consists of a maximum matching in $G$
    and a perfect matching in $Z^+ \times V^-$.
  \item
    $Z^+$ is included in a single strongly connected component of $G'(M') = G' + \overline{M'}$,
    which is a unique source component, and hence $s(G'(M')) = 1$.
  \item
    If $V_\infty \neq \emptyset$, then $Z^+ \subseteq V_\infty^+$, and hence $G' - V_\infty = G - V_\infty$.
    In particular, $G'[V_0] = G[V_0]$.
  \end{itemize}
\end{observation}

By Observation~\ref{obs:balanced},
a maximum matching $M' \subseteq E$ in $G'$
consists of a maximum matching in $G$ and a perfect matching in $G'[Z^+ \cup V^-]$.
Hence, we can find a maximum matching in $G'$ in ${\rm O}(nm)$ time
just by doing so in $G$ and adding an arbitrary perfect matching
between the exposed vertices in $G'[Z^+ \cup V^-]$.
In addition, since $Z^+$ is included in a single strongly connected component of $G'(M')$,
we can regard $Z^+$ as a single vertex in computing the strongly connected component of $G'(M')$.
This makes it possible to obtain the strongly connected components of $G'(M')$
in ${\rm O}(n + m)$ time, which concludes that Step 0 can be done in ${\rm O}(nm)$ time.

Since Step 4 is also done in ${\rm O}(n + m)$ time by the same argument,
it suffices to bound the running time of Step 2 by ${\rm O}(nm)$.
If $V_\infty = \emptyset$, then we do not reach Step 2.
Otherwise, by Observation~\ref{obs:balanced},
we see $Z^+ \subseteq V_\infty^+$ and $G'[V_0] = G[V_0]$.
Hence, one can find an eligible perfect matching in $G'[V_0]$ in ${\rm O}(nm)$ time by Procedure EPM.
In addition, since no edge enters $V_\infty$ in $G'(M')$ by Observation~\ref{obs:source_sink},
the strongly connected component including $Z^+$ is a unique source component
also in $G'(M')[V_\infty]$, and hence $s(G'(M')[V_\infty]) = 1$.
This condition does not depend on the choice of $M'$,
which means that all the perfect matchings in $G'[V_\infty]$ is eligible.
Hence, we do not need to use Procedure EPM for finding an eligible perfect matching in $G'[V_\infty]$,
which concludes that Step 2 can be done in ${\rm O}(nm)$ time.

\section{Finding an Optimal Subpartition}\label{sec:app1}
In our min-max duality theorems (Theorems \ref{thm:duality} and \ref{thm:duality_unbalanced}),
we take the maximum of
\begin{align*}
  \tau_G(\cX) = \sum_{X \in \cX}\left(|X| - |\Gamma_G(X)| + 1\right),
\end{align*}
over all (proper) subpartitions $\cX$ of $V^+$ (and of $V^-$).
This situation is generalized as follows.
Given an intersecting supermodular function $g \colon 2^S \to \RR$ with $g(\emptyset) = 0$
over some finite set $S$, find a (proper) subpartition $\cX$ of $S$ that maximizes
\[\tau_g(\cX) := \sum_{X \in \cX} g(X).\]
With the aid of efficient submodular function minimization algorithms,
one can find such a maximizer $\cX$ in polynomial time as follows.

Let $Q(g)$ be the associated polyhedron defined by
\[Q(g) = \{\, z \mid z \in \RR_{\geq 0}^S,\ z(X) \geq g(X) \ (\forall X \subseteq S) \,\},\]
where $z(X) := \sum_{v \in X} z_v$.
Note that for any $z \in Q(g)$ and any subpartition $\cX$ of $S$,
we have $z(S) \geq \tau_g(\cX)$.
Consider the following algorithm.
\begin{description}
  \setlength{\itemsep}{.5mm}
\item[Step 0.]
  Take an arbitrary vector $z \in Q(g)$.
  Set $U \leftarrow S$ and $j \leftarrow 0$.

\item[Step 1.]
  While $z(U) > 0$ do the following.
  
  \setlength{\leftskip}{5mm}

\item[Step 1.1.]
  Select an arbitrary element $v \in U$ with $z_v > 0$.

\item[Step 1.2.]
  Compute $\alpha := \min\{\, z(X) - g(X) \mid v \in X \subseteq U \,\}$.
  If $\alpha < z_v$, then $j \leftarrow j + 1$,
  let $X_j$ be a unique maximal minimizer,
  $z_v \leftarrow z_v - \alpha$, and $U \leftarrow U \setminus X_j$.
  Otherwise, $z(v) \leftarrow 0$.
\end{description}\vspace{1mm}

Let $k$ be the value of $j$ at the end of this algorithm.
Then, $\cX := \{X_1, X_2, \ldots, X_k\}$ is a subpartition of $S$.
The vector $z$ remains in $Q(g)$ throughout the algorithm.
At the end of the algorithm, we have $z(X_j) = g(X_j)$ for every $j \in [k]$, and $z(U) = 0$.
Thus we obtain $z(S) = \tau_g(\cX)$, which implies that $\cX$ maximizes $\tau_g(\cX)$
over all subpartitions of $S$.

In order to find an optimal ``proper'' subpartition of $S$,
one can use the above algorithm to obtain an optimal subpartition of $S \setminus \{v\}$ for each $v \in S$,
and take the best among all the obtained subpartitions.
\end{appendix}

\end{document}